\documentclass[a4paper,12pt,twoside]{article}

\usepackage{latexsym,epic,eepic,float}
\usepackage{color,colordvi}
\usepackage{shadow}
\usepackage{amsmath,amssymb,amsfonts,stackrel}
\usepackage{amsmath}
\usepackage{epsfig}
\usepackage{lscape}
\usepackage{rotating}
\usepackage{subfigure}
\usepackage{multicol}
\usepackage{tikz}
\usepackage{graphicx}
\usepackage{bbold}

\usepackage[T1]{fontenc}
\usepackage[utf8]{inputenc}
\usepackage{authblk}

\newcommand{\subf}[2]{%
  {\small\begin{tabular}[t]{@{}c@{}}
  #1\\#2
  \end{tabular}}%
  }

\newtheorem{theorem}{Theorem}[section]
\newtheorem{theo}[theorem]{Theorem}
\newtheorem{lem}[theorem]{Lemma}

\newtheorem{prop}[theorem]{Proposition}

\newtheorem{Ex}[theorem]{Example}

\newcommand{\exit}{{\mbox{\, \vspace{3mm}}} \hfill\mbox{$\square$}}
\rm 

\numberwithin{equation}{section}

\title{Optimal valuation of American callable credit default swaps under drawdown of L\'evy insurance risk process}

\author{Z. Palmowski
\thanks{Faculty of Pure and Applied Mathematics, Wroc\l aw University of Science and Technology, ul. Wyb. Wyspia\'nskiego 27, 50-370 Wroc\l aw, Poland,
email: zbigniew.palmowski@gmail.com}
\thanks{This work is partially supported by Polish National Science Centre Grant
No. 2016/23/B/HS4/00566 (2017-2020).}
, B.A. Surya
\thanks{School of Mathematics and Statistics, Victoria University of Wellington, New Zealand, email: budhi.surya@msor.vuw.ac.nz}
\thanks{The author acknowledges financial support through PBRF research grant No. 220859} }

\date{14 March 2020}

\begin{document}
\maketitle \pagestyle{myheadings} \markboth{Z. Palmowski, B.A. Surya} {American callable credit default swaps under drawdown}
\begin{abstract}
This paper discusses the valuation of credit default swaps, where default is announced when the reference asset price has gone below certain level from the last record maximum, also known as the high-water mark or drawdown. We assume that the protection buyer pays premium at fixed rate when the asset price is above a pre-specified level and continuously pays whenever the price increases. This payment scheme is in favour of the buyer as she only pays the premium when the market is in good condition for the protection against financial downturn. Under this framework, we look at an embedded option which gives the issuer an opportunity to call back the contract to a new one with reduced premium payment rate and slightly lower default coverage subject to paying a certain cost. We assume that the buyer is risk neutral investor trying to maximize the expected monetary value of the option over a class of stopping time. We discuss optimal solution to the stopping problem when the source of uncertainty of the asset price is modelled by L\'evy process with only downward jumps. Using recent development in excursion theory of L\'evy process, the results are given explicitly in terms of scale function of the L\'evy process. Furthermore, the value function of the stopping problem is shown to satisfy continuous and smooth pasting conditions regardless of regularity of the sample paths of the L\'evy process. Optimality and uniqueness of the solution are established using martingale approach for drawdown process and convexity of the scale function under Esscher transform of measure.  Some numerical examples are discussed to illustrate the main results.

\medskip

\textbf{Keywords}: L\'evy process; drawdown; credit risk; credit default swaps

\end{abstract}

\section{Credit default swaps}

Credit default swaps is one of the financial instruments that provide an insurance which may be used to offset financial loss due to a credit event experienced by a borrower. In the credit event, the borrower may not be able to fully meet its obligation to payback the required interest on the debt or principal on time.

Over the past decades, some discussions have been developed on risk protection mechanism against financial asset's outperformance over its last record maximum, or high-water mark also known as the drawdown, which may affect towards fund managers' compensation; see, among others, Agarwal et al. \cite{Agarwal} and Goetzmann et al. \cite{Goetzmann} for details. Unlike the classical Gerber-Shiu ruin theory, see for e.g. Kyprianou \cite{Kyprianou2013} and the literature therein, where default is announced when the underlying asset price process crosses below a threshold, default under drawdown is triggered when the price process has gone below a certain level from its previous maximum. Analysis of the Parisian type of default with reference to the last record maximum is given in Surya \cite{Surya}. Such risk may be protected against using an insurance. In their recent works, Zhang et al. \cite{Zhang} and Palmowski and Tumilewicz \cite{Palmowski2018} discussed fair valuation and design of such insurance. It is worth noticing that in both papers, the insurance premium paid by the protection buyer is not contingent on the insured underlying assets.

When the credit event is due to financial restructuring, which occurs when the financial liabilities of the borrower are changed, the restructuring could change the debt contract's subordination, reducing its debt priority in the event of default, and affect the pricing of credit default swaps in the market. See Altman et al. \cite{Altman} for details. In their empirical work, Berndt \cite{Berndt} found that when default swap rates without restructuring increase, the increase in restructuring premium for CDS is higher for low-credit-quality firms than for high-credit-quality firms. That is restructuring CDS premium depends on firm's specific balance-sheet. 

In the framework of default by drawdown, we extend the credit default swaps problem considered in e.g. Leung and Yamazaki \cite{Leung}, by allowing the protection buyer to pay the premium whenever the price of reference asset is above certain threshold and continues to pay when the price is increasing, and pays nothing otherwise. This payment scheme is in favours of the protection buyer whereby she/he only pays when economy is doing good and stop paying in reverse condition. In this regard and by work of \cite{Berndt}, we let default premium depend on the market condition. To our knowledge, this premium payment scheme is relatively new.

The source of uncertainty in the reference asset is modeled by exponential L\'evy process. For this purpose, let $X=\{X_t:t\geq 0\}$ be a L\'evy process with downward jumps defined on filtered probability space $(\Omega, \mathcal{F}, \{\mathcal{F}_t:t\geq 0\},\mathbb{P})$, where $\mathcal{F}_t$ is the natural filtration of $X$ satisfying the usual conditions of right-continuity and completeness. Denote by $\{\mathbb{P}_x, x\in\mathbb{R}\}$ the family of probability measure corresponding to a translation of $X$ such that $X_0=x$, with $\mathbb{P}=\mathbb{P}_0$.
We assume that $0$ is irregular for $(-\infty, 0)$. In particular, we exclude in this paper a downward subordinator case.
The classical example that we have in mind is a Cram\'er-Lundberg risk process with positive drift $c>0$.

 The L\'evy-It\^o sample paths decomposition of the L\'evy process is given by
\begin{equation}\label{eq:LevyIto}
\begin{split}
X_t=\mu t + \sigma B_t &+\int_0^t\int_{\{x<-1\}} x\nu(dx,ds) \\
&+ \int_0^t\int_{\{-1\leq x <0\}} x\big(\nu(dx,ds)-\Pi(dx)ds\big),
\end{split}
\end{equation}
where $\mu\in\mathbb{R}$, $\sigma\geq0$ and $(B_t)_{t\geq0}$ is standard Brownian motion, whilst $\nu(dx,dt)$ denotes the Poisson random measure associated with the jumps process $\Delta X_t:=X_t-X_{t-}$ of $X$. This Poisson random measure has compensator given by $\Pi(dx)dt$, where $\Pi$ is the L\'evy measure satisfying the integrability condition:
\begin{equation}\label{eq:intcond}
\int_{-\infty}^0 (1\wedge x^2)\Pi(dx)<\infty.
\end{equation}

Due to the absence of positive jumps, it is therefore sensible to define
\begin{equation}\label{eq:exponent}
\psi(\lambda)=\frac{1}{t}\log\mathbb{E}\big\{e^{\lambda
X_{t}}\big\}=\mu\lambda
+\frac{1}{2}\sigma^{2}\lambda^{2}+\int_{(-\infty,0)}\big(e^{\lambda
x}-1-\lambda x\mathbf{1}_{\{x>-1\}}\big)\Pi(dx),
\end{equation}
the Laplace exponent of $X$, which is analytic on ($\mathfrak{Im}(\lambda)\leq 0$). It is easily
shown that $\psi$ is zero at the origin, tends to infinity at
infinity and is strictly convex.
In the case of Cram\'er-Lundberg model
$c=\mu- \int_{(-\infty,0)}x\mathbf{1}_{\{x>-1\}}\Pi(dx)$
and $\Pi(dx)=\beta F(dx)$ for Poisson arrival intensity $\beta>0$ and distribution function
of downward jumps $F$.
We denote by
$\Phi:[0,\infty)\rightarrow [0,\infty)$ the right continuous inverse
of $\psi(\lambda)$, so that
\begin{equation*}
\Phi(\theta)=\sup\{p>0:\psi(p)=\theta\} \quad \textrm{and} \quad
\psi(\Phi(\lambda))=\lambda \quad \text{for all} \quad \lambda \geq
0.
\end{equation*}
We refer to Ch. VI in Bertoin \cite{Bertoin} or Ch. 2 in Kyprianou \cite{Kyprianou} for details.

Next, we denote by $\overline{X}_t=\sup_{0\leq s \leq t} X_s$ the running maximum of $X$ up to time $t$ and assume that  from some arbitrary prior point of reference in time $X$ has the current maximum $y\geq x$. Define $S_t=\overline{X}_t\vee y$, where $a\vee b=\max\{a,b\}$ and the reflected process $Y_t=S_t-X_t$. Recall that the process $Y=\{Y_t:t\geq 0\}$ possesses strong Markov property. Furthermore, we alter slightly our notation for the probability measure $\mathbb{P}_{x,y}$ under which at time zero $X$ has the current maximum $y\geq x$ and position $x\in\mathbb{R}$, and we simply write $\mathbb{P}_{\vert y}:=\mathbb{P}_{0,y}$ to denote the law of $Y$ under which $Y_0=y$, and use the notation $\mathbb{E}_x$,  $\mathbb{E}_{x,y}$ and $\mathbb{E}_{\vert y}$ to define the corresponding expectation operator to the above probability measures.

Following \cite{Palmowski2018, Surya, Zhang}, default is announced as soon as the underlying L\'evy process has gone below a fixed level $b>0$ from its last record maximum,
\begin{equation*}
\tau_b^+=\inf\{t>0: Y_t>b\}.
\end{equation*}
We assume that an optimal default level $b$ has been chosen endogenously by optimizing the CDS issuing firm's equity/capital structure such as discussed, among others, in Leland and Toft \cite{Leland} and Kyprianou and Surya \cite{Kyprianou2007}. Under a $T-$year CDS on a unit face value, the protection buyer continuously pays
with rate $pS_t$ for fixed $p>0$ whenever the reference asset price is above current maximum and increasing over time until default occurs or maturity $T$, whichever is sooner.
For the Cram\'er-Lundberg risk process with the drift $c>0$ on account that $c\mathbf{1}_{\{Y_t=0\}}dt\stackrel{d}{=}dS_t$ under $\mathbb{P}_{\vert 0}$ we have that
$pc$ is a fixed CDS premium.
If default occurs prior to $T$, the buyer will receive the default payment $\alpha:=1-R$ at default time $\tau_b^+$, where $R$ is the assumed constant recovery rate (typically $40\%$). From the buyer's perspective, the expected discounted payoff of CDS is
\begin{equation}\label{eq:payoff}
\overline{C}_T(x,y,b; p,\alpha)=\mathbb{E}_{x, y}\Big[-\int_0^{\tau_b^+ \wedge T} e^{-rt} pdS_t + \alpha e^{-r\tau_b^+}\mathbf{1}_{\{\tau_b^+\leq T\}}\Big],
\end{equation}
where $r>0$ is a fixed risk-free interest rate. The quantity $\overline{C}_T(x,y,b; p,\alpha)$ can be viewed as the market price for the buyer to enter (or long) a CDS contract with agreed premium $p>0$, promised payment upon default $\alpha$ and the maturity $T$.

On the opposite side of the trade, the protection seller's expected cash flow is $-\overline{C}_T(x,y,b; p,\alpha)=\overline{C}_T(x,y,b; -p,-\alpha).$ In practice, the CDS premium $\overline{p}$ is set as such that $\overline{C}_T(x,y,b; \overline{p},\alpha) =0$, yielding zero expected cash flows for both parties. Following (\ref{eq:payoff}) it is straightforward to show that the credit spreads $\overline{p}$ is

\begin{equation}\label{eq:spread}
\overline{p}=\frac{\alpha \mathbb{E}_{x,y}\big[e^{-r\tau_b^+}\mathbf{1}_{\{\tau_b^+\leq T\}}\big]}{
\mathbb{E}_{x,y}\big[\int_0^{\tau_b^+ \wedge T} e^{-rt} dS_t\big]}.
\end{equation}
Recall that the two quantities $\mathbb{E}_{x,y}\big[\int_0^{\tau_b^+ \wedge T} e^{-rt} dS_t\big]$
and $\mathbb{E}_{x,y}\big[e^{-r\tau_b^+}\mathbf{1}_{\{\tau_b^+\leq T\}}\big]$ are not in general available in explicit form. However, following Theorem 1 in Avram et al. \cite{Avram2004} their Laplace transforms on maturity $T$ are given in terms of the so-called scale function $W^{(u)}(x)$ whose Laplace transform is defined by
\begin{equation}\label{eq:scale}
\int_0^{\infty} e^{-\lambda x} W^{(u)}(x) dx =\frac{1}{\psi(\lambda)-u}, \quad \textrm{for $\lambda>\Phi(u)$, $u>0$,}
\end{equation}
with $W^{(u)}(x)=0$ for $x<0$, which is increasing and is continuously differentiable for $x>0$ when $X$ has paths of unbounded variation and bounded variation with the L\'evy measure $\Pi$ has no atoms which is assumed from now on.
See \cite{Kyprianou} for details. We write $W(x)=W^{(0)}(x)$.

Define the function $Z^{(u)}(x):=1+u\int_{0}^x W^{(u)}(z) dz.$ Then, following \cite{Avram2004},
\begin{equation}\label{eq:identity0}
\mathbb{E}_{x,y}\big[e^{-\lambda \tau_b^+}\big]=Z^{(\lambda)}(b-z)-\lambda\frac{W^{(\lambda)}(b)}{W^{(\lambda)\prime}(b)}W^{(\lambda)}(b-z),
\end{equation}
where $z=y-x\geq 0$.

Similarly, the Laplace transform of the discounted payoff $\int_0^\infty e^{-\beta T}\mathbb{E}_{x,y}\big[\int_0^{\tau_b^+ \wedge T} e^{-rt}  dS_t\big]dT$ could be derived from the identity \eqref{eq:identity1} proved in \cite{Avram2007} and applied for the exponentially killed L\'evy process $X$ with intensity $\beta$, that is by taking $q=r+\beta$ there.
Recall that due to spatial homogeneity of the sample paths of the L\'evy process $X$, it is therefore sufficient to consider the valuation (\ref{eq:payoff}) under the measure $\mathbb{P}_{\vert y}.$ To simplify analysis from now on
we focus only on the perpetual counterpart of (\ref{eq:payoff}):
\begin{equation}\label{eq:perpayoff}
\overline{C}_{\infty}(y,b; p,\alpha)=\mathbb{E}_{\vert y}\Big[-\int_0^{\tau_b^+} e^{-rt} p dS_t + \alpha e^{-r\tau_b^+}\mathbf{1}_{\{\tau_b^+ < \infty\}}\Big].
\end{equation}
The main goal of this work is pricing American callable credit default swaps.

The organization of this paper is as follows. Section 2 discusses American callable credit default swaps in further details and presents some preliminary results required to solve the corresponding  optimal stopping problem. Solution of the stopping problem is presented in Section 3. Optimality and uniqueness of the solution are discussed in Section 4. Some numerical examples are given in Section 5 to exemplify the main results. Section 6 concludes this paper.

\section{American callable credit default swaps}

Following \cite{Leung} we consider a credit default swap contract under drawdown which allows the issuer to call/replace the CDS contract once with reduced premium payment rate and lower default coverage. By doing so, the issuer is subject to paying a certain amount of fee $\gamma$. For convenient, we consider perpetual case.

At any time prior to default, the buyer can select a stopping time $\theta$ to switch to a new contract with a new premium payment rate $\widehat{p}<p$ and default coverage $\widehat{\alpha}<\alpha$ subject to a fee payment $\gamma$.  The default payment then changes from $\alpha$ to $\widehat{\alpha}=q\alpha$, $q<1$, after the exercise time $\theta$ of the new contract. Given that the buyer is risk neutral, she/he is interested in finding optimal stopping time $\theta$ to switch the CDS contract so as to maximizes the net expected cash flow:
\begin{align}
V_b(y; p, \widehat{p},\alpha,\widehat{\alpha},\gamma)
&=\sup_{\theta\in\mathcal{T}_{[0,\infty)}}\mathbb{E}_{\vert y}\Big[-\int_0^{\theta\wedge \tau_b^+} e^{-rt}pdS_t + e^{-r \tau_b^+}\big(\widehat{\alpha} \mathbf{1}_{\{\theta \leq \tau_b^+\}} + \alpha  \mathbf{1}_{\{\theta > \tau_b^+\}}\big) \nonumber \\
&\hspace{2cm}-\mathbf{1}_{\{\theta \leq \tau_b^+\}}\Big(\int_{\theta}^{\tau_b^+} e^{-rt} \widehat{p} dS_t + e^{-r\theta}\gamma \Big)\Big], \label{eq:perpayoff2}
\end{align}
where $\mathcal{T}_{[a,b)}$ denotes the class of $\mathcal{F}_t-$stopping times taking values in $[a,b), a\geq0.$ Notice that the pay-off structure (\ref{eq:perpayoff2}) is slightly different from that of \cite{Leung}.
\begin{prop}\label{representation}
Define $\widetilde{p}=\widehat{p}-p$ and $\widetilde{\alpha}=\widehat{\alpha}-\alpha.$ Then, following (\ref{eq:perpayoff2}),
\begin{equation}\label{eq:perpost1}
V_b(y; p, \widehat{p},\alpha,\widehat{\alpha},\gamma)=\overline{C}_{\infty}(y,b; p,\alpha) + \mathcal{V}_b(y;\widetilde{p},\widetilde{\alpha},\gamma),
\end{equation}
where $\mathcal{V}_b(y;\widetilde{p},\widetilde{\alpha},\gamma)$, for $y\in(0,b)$, is the value function of the stopping problem,
\begin{equation}\label{eq:perpost2}
\mathcal{V}_b(y;\widetilde{p},\widetilde{\alpha},\gamma)=\sup_{\theta\in\mathcal{T}_{[0,\infty)}}\mathbb{E}_{\vert y}\Big[\mathbf{1}_{\{\theta \leq \tau_b^+\}}\Big(-\int_{\theta}^{\tau_b^+}e^{-rt}\widetilde{p} dS_t  + e^{-r\tau_b^+}\widetilde{\alpha} - e^{-r\theta} \gamma\Big)\Big].
\end{equation}
\end{prop}
{\it Proof} By rearranging the payoff structure of the value function (\ref{eq:perpayoff2}),
\begin{eqnarray*}
&&-\int_0^{\theta\wedge \tau_b^+} e^{-rt}p dS_t + e^{-r \tau_b^+}\big(\widehat{\alpha} \mathbf{1}_{\{\theta \leq \tau_b^+\}} + \alpha  \mathbf{1}_{\{\theta > \tau_b^+\}}\big) \notag \\
&&\hspace{4.5cm}-\mathbf{1}_{\{\theta \leq \tau_b^+\}}\Big(\int_{\theta}^{\tau_b^+} e^{-rt} \widehat{p} dS_t + e^{-r\theta}\gamma \Big) \notag\\
&=&\mathbf{1}_{\{\theta \leq \tau_b^+\}}\Big[-\int_{\theta}^{\tau_b^+}e^{-rt}\widetilde{p}dS_t + e^{-r\tau_b^+}\widetilde{\alpha} - e^{-r\theta} \gamma \Big] \notag\\
&&\hspace{4.5cm}-\int_0^{\tau_b^+}e^{-rt} p dS_t + e^{-r\tau_b^+}\alpha. \quad \exit
\end{eqnarray*}
This produces the assertion of the proposition.
\exit
%\end{proof}

We observe following (\ref{eq:perpost1}) that the introduction of the embedded option increases the buyer's expected monetary value by the value function $\mathcal{V}_b(y;\widetilde{p},\widetilde{\alpha},\gamma)$.

We will show in the next section that an optimal stopping to the problem (\ref{eq:perpost2}) takes the form of first passage below a level of the drawdown process $Y$
\begin{equation}\label{eq:fptbelow}
\tau_h^-=\inf\{t>0: Y_t <h\}, \; \textrm{with $h>0$, under $\mathbb{P}_{\vert y}$}.
\end{equation}
To this end, we denote by $h^{\star}\in(0,b)$ the largest root, whenever it exists, of
\begin{equation}\label{eq:bigroot}
\widetilde{\alpha} Z^{(r)}(b-h)-r \widetilde{\alpha}\frac{\big(W^{(r)}(b-h)\big)^2}{W^{(r)\prime}(b-h)}-\gamma=0.
\end{equation}
In the section below we give some preliminary results to solve (\ref{eq:perpost2}).

\begin{figure}[t!]
\centering
\begin{tabular}{cc}
%\hline
\subf{\includegraphics[width=72.5mm]{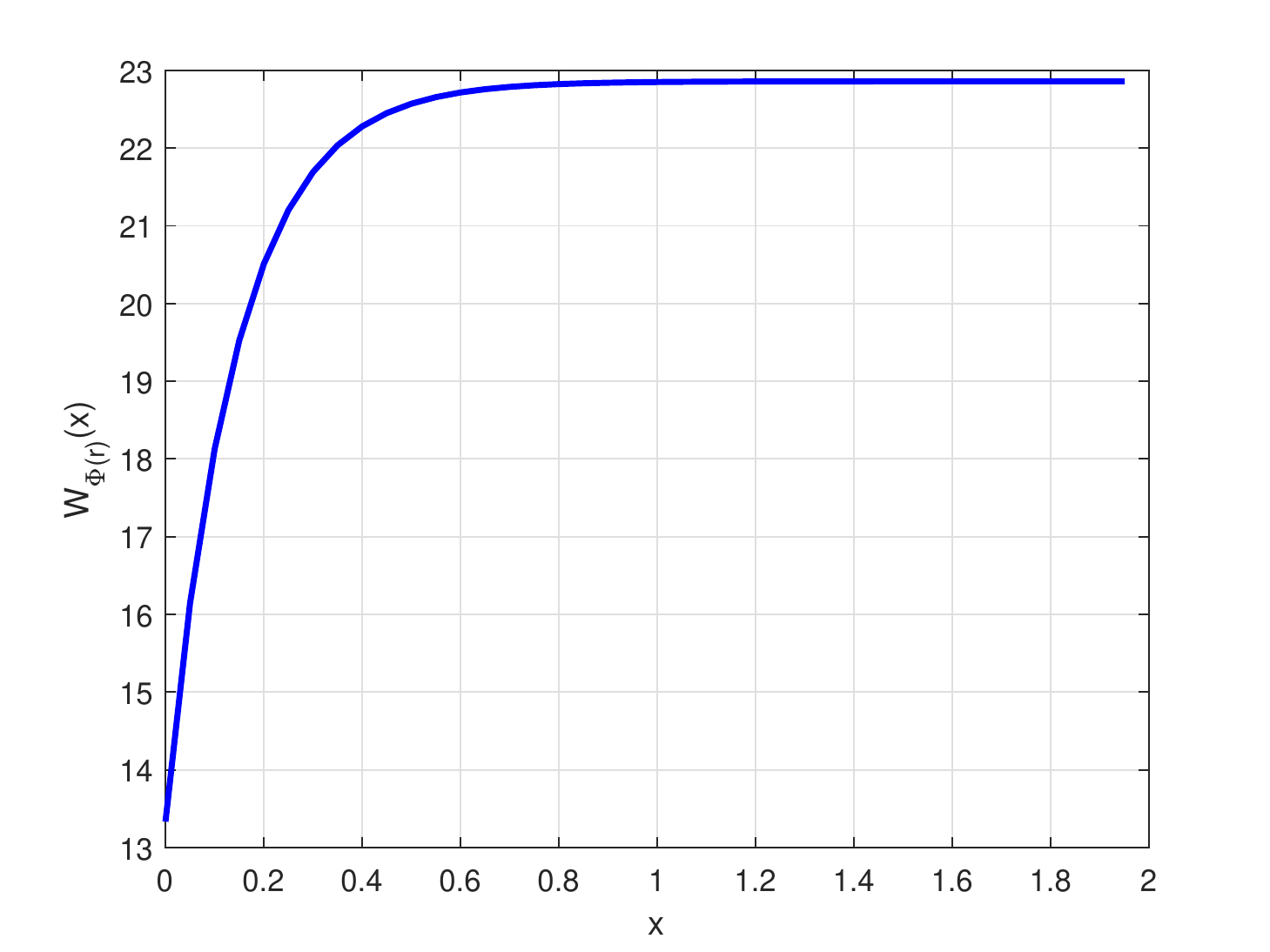}}
       {the case $\mu=0.075$ and $\sigma=0$.}

&
\subf{\includegraphics[width=72.5mm]{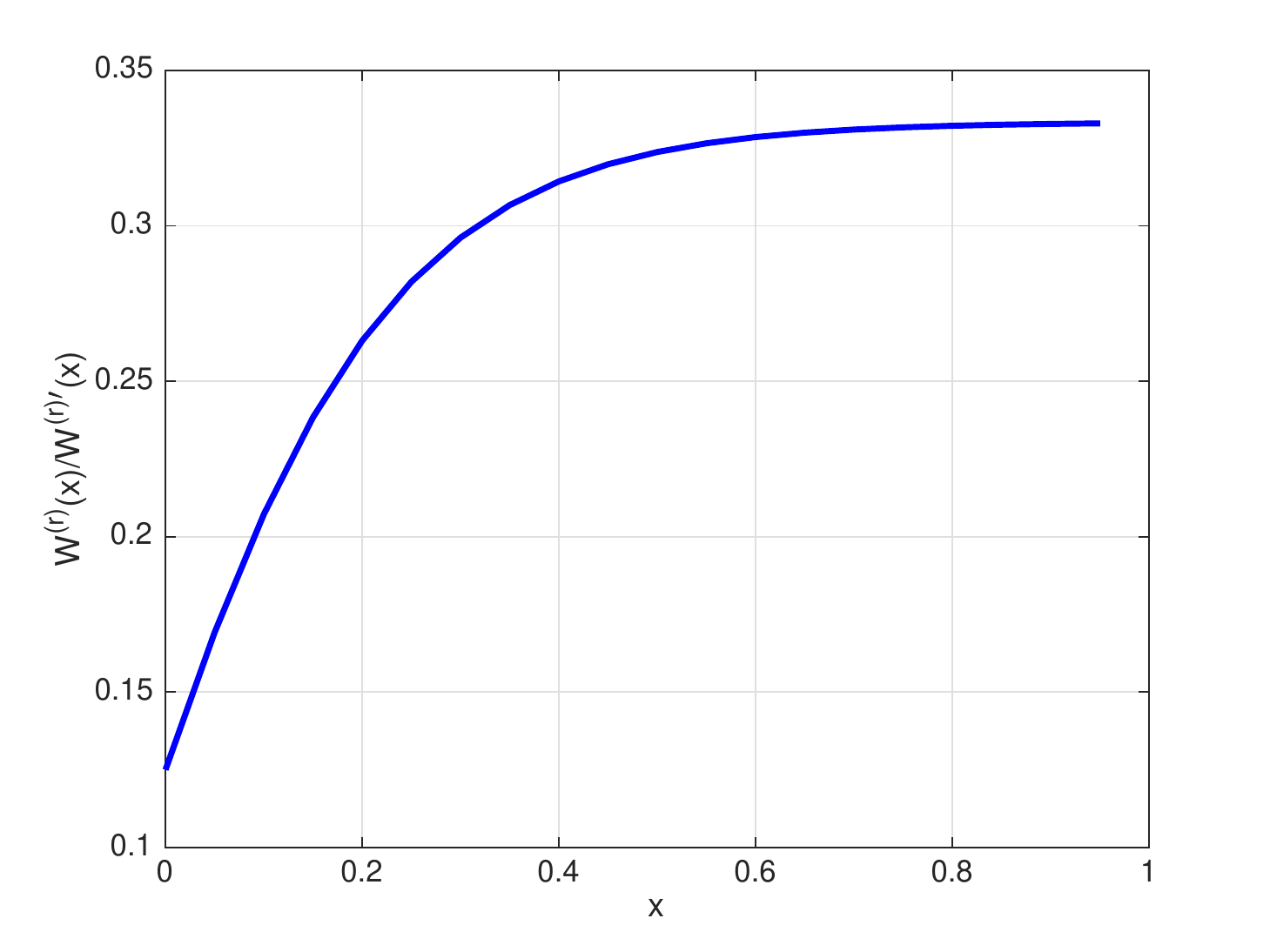}}
     {the case  $\mu=0.075$ and $\sigma=0$.}
\\
%\hline
\subf{\includegraphics[width=72.5mm]{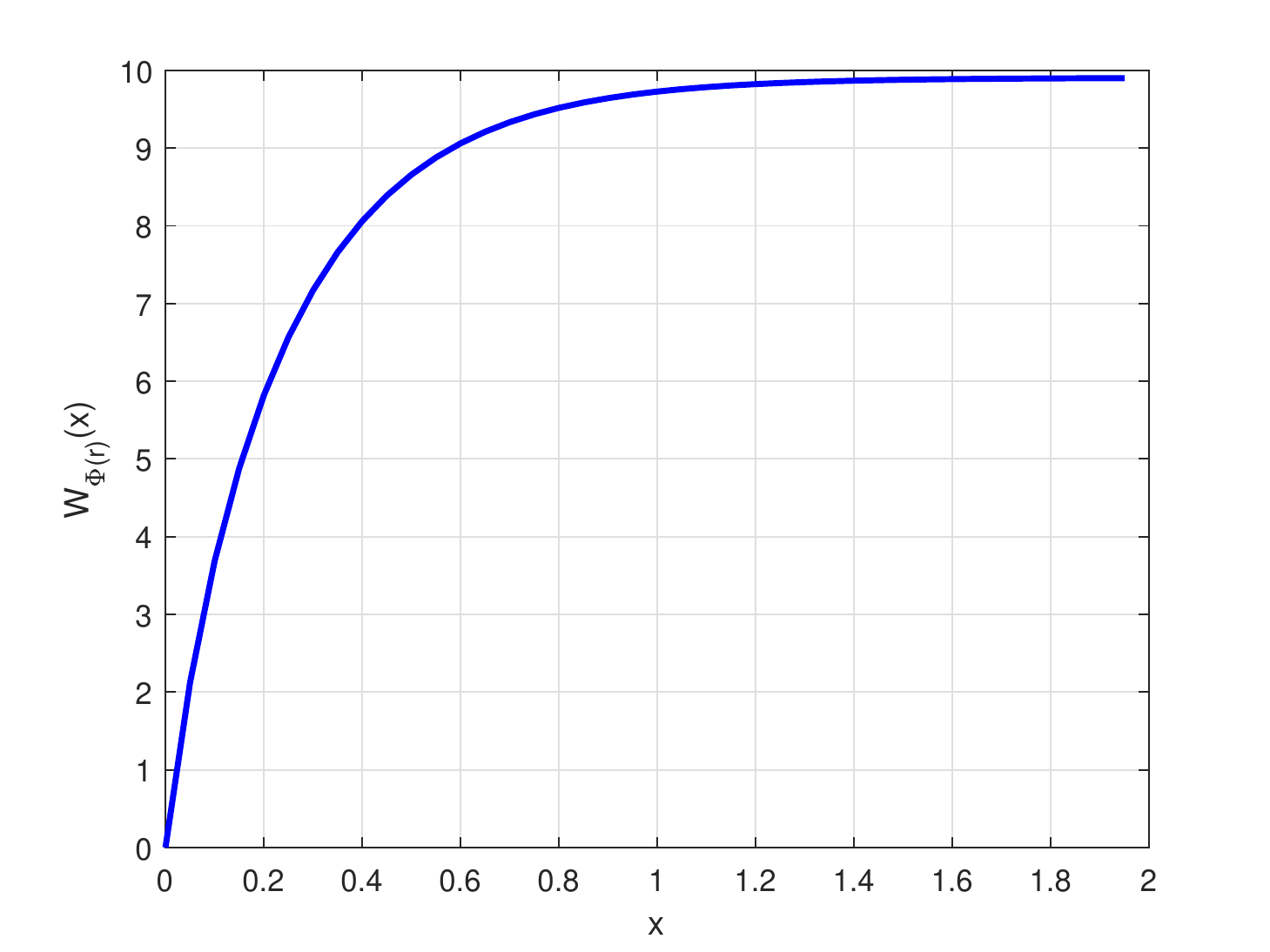}}
       {the case  $\mu=0.075$ and $\sigma=0.2$.}
&
\subf{\includegraphics[width=72.5mm]{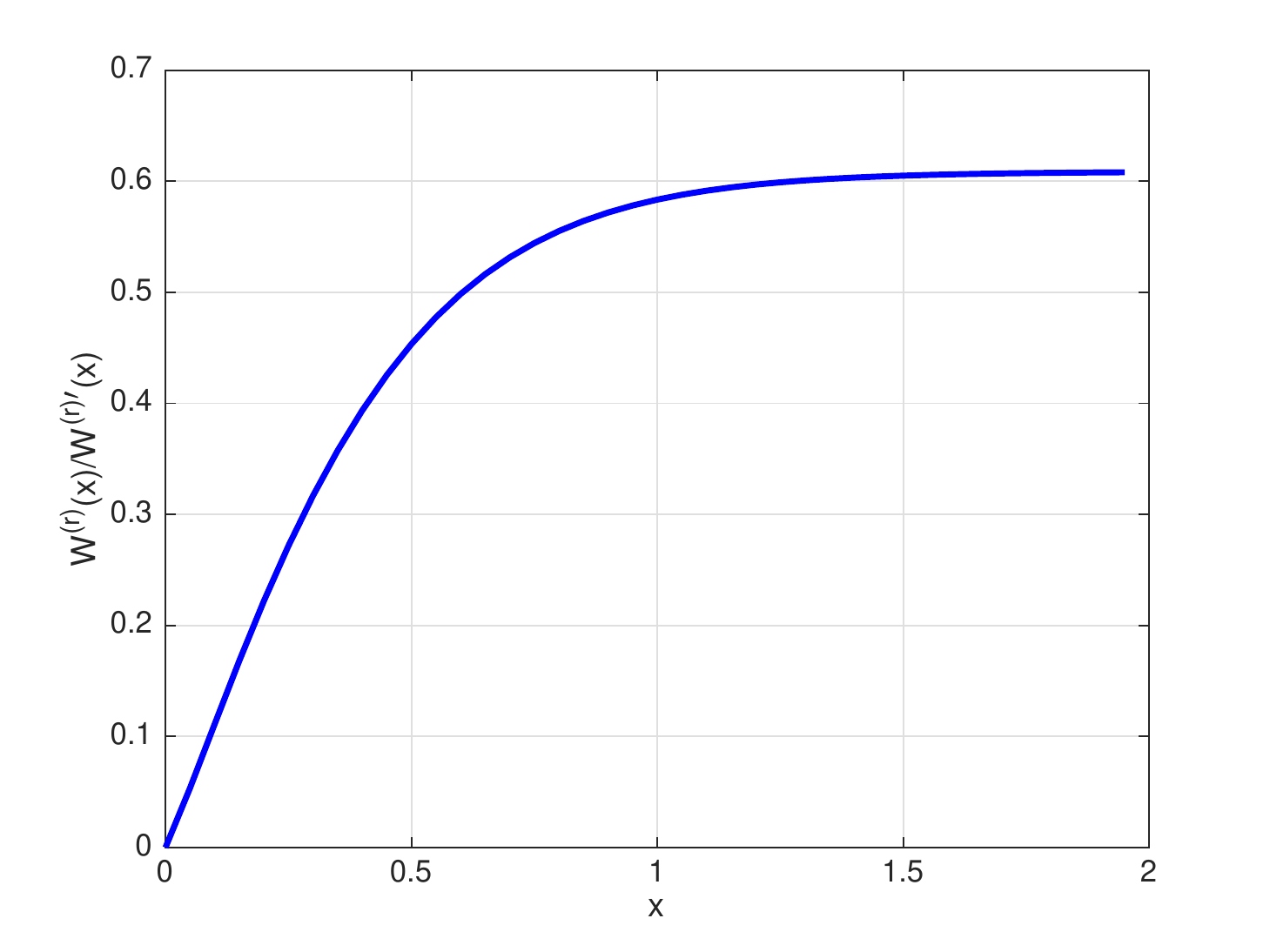}}
     {the case  $\mu=0.075$ and $\sigma=0.2$.}
\\
%\hline
\end{tabular}
\caption{Plots of $W_{\Phi(r)}(x)$ (\ref{eq:ScaleF2}) and $W^{(r)}(x)/W^{(r)\prime}(x)$, with $r=0.1$, for jump-diffusion process with $\psi(\lambda)=\mu\lambda + \frac{\sigma^2}{2}\lambda^2 -\frac{a\lambda}{\lambda+c}$ for $a=0.5$ and $c=9$ .}\label{fig:scale}
\end{figure}

\subsection{Preliminaries}
The decomposition of the value function $V_b(y; p, \widehat{p},\alpha,\widehat{\alpha},\gamma)$(\ref{eq:perpost1})-(\ref{eq:perpost2}) can be represented in terms of the scale function $W^{(u)}(x)$ (\ref{eq:scale}). The scale function $W^{(u)}(x)$ plays an important role in getting semi-explicit solution to the two-sided exit problem of L\'evy process $X$ as shown by the following identity. Let $T_b^+$ and $T_b^-$ be respectively the first entrance time of $X$ into $(b,\infty)$ and $(-\infty,-b)$, for $b>0$, defined by the $\mathcal{F}_t-$stopping times $T_{-b}^-=\inf\{t\geq 0: X_t<-b\}$ and $T_b^+=\inf\{t\geq 0: X_t>b\}$. Under the measure $\mathbb{P}_x$, the identity concerning the first exit of $X$ above level $b\geq x$ before first passage of $X$ below zero is given by
\begin{equation}\label{eq:fpt}
\mathbb{E}_x\big[e^{-uT_b^+} \mathbf{1}_{\{T_b^+<T_0^-\}}\big]=\frac{W^{(u)}(x)}{W^{(u)}(b)}.
\end{equation}

The key to obtaining solution to (\ref{eq:perpost1})-(\ref{eq:perpost2}) is given by the following identities.
We remind that we consider only the case $X$ has either paths of unbounded or bounded variations with absolutely continuous jumps.
The last case is equivalent to Cram\'er-Lundberg risk process $X$ with positive drift $c>0$ (since we excluded downward subordinator from our considerations).
\begin{prop}\label{prop:prop1}
For given $q,b>0$ the following identities hold $\forall y\in(0,b)$:
\begin{align}
\mathbb{E}_{\vert y}\Big[\int_0^{\tau_b^+}e^{-qt}dS_t\Big]&=\frac{W^{(q)}(b-y)}{W^{(q)\prime}(b)},\label{eq:identity1}\\
\mathbb{E}_{\vert y}\big[e^{-q\tau_b^+}\big]&= Z^{(q)}(b-y)-q\frac{W^{(q)}(b)}{W^{(q)\prime}(b)}  W^{(q)}(b-y). \label{eq:identity2}
\end{align}

\end{prop}
Note that the identity (\ref{eq:identity1}) is the dual version of the identity (3.12) of  \cite{Avram2007} under dividend controlled risk process $U_t$, which is equal in distribution under $\mathbb{P}_{\vert 0}$ to the dual of the drawdown process $Y_t$, i.e., $U_t=-Y_t$.

It is also worth noting that under a new change of measure $\mathbb{P}_x^{\nu}$ defined by the Esscher transform $d\mathbb{P}_x^{\nu}/d\mathbb{P}_x=e^{\nu(X_t-x)-\psi(\nu)t}$, $(X,\mathbb{P}_x^{\nu})$ is a spectrally negative L\'evy process. Under the new measure, it is straightforward to check by taking Laplace transform on both sides that $W^{(u)}(x)=e^{\Phi(u)x}W_{\Phi(u)}(x)$, where $W_{\Phi(u)}(x)=W^{(0)}_{\Phi(u)}(x)$ is the scale function under $\mathbb{P}^{\Phi(u)}$.
Now from Lemma 8.2 of \cite{Kyprianou} it follows that $\frac{W_{\Phi(u)}(x)}{W^{\prime}_{\Phi(u)}(x)}$ is
monotone increasing function in $x$ as it is a reciprocal of the rate of excursions larger than $x$. Indeed, the identity (\ref{eq:identity1}) justifies this assertion under $\mathbb{P}_{\vert 0}$.
In this way we have the following result that will be used later to establish optimality and uniqueness of the solution to the stopping problem (\ref{eq:perpost2}).
\begin{lem}
For $u\geq 0$, $W^{(u)}(x)/W^{(u)\prime}(x)$ is monotone increasing in $x$, i.e.,
\begin{align}\label{eq:monotone}
\frac{d}{dx} \Big(\frac{W^{(u)}(x)}{W^{(u)\prime}(x)}\Big)\; > \;0, \; \forall \; x\geq 0, \; \textrm{and is bounded above by $1/\Phi(u)$.}
\end{align}
\end{lem}

\begin{Ex}\label{ex:Ex1}\rm
Consider one-sided jump-diffusion process $X$ with $\psi(\lambda)=\mu\lambda+\frac{\sigma^2}{2}\lambda^2-\frac{a\lambda}{\lambda+c}$ for all $\lambda\in\mathbb{R}$ s.t. $\lambda\neq -c$. It is known, see e.g. \cite{Kyprianou}, that for $u>0$,
\begin{align}\label{eq:ScaleF}
W^{(u)}(x)=\frac{e^{-\xi_2 x}}{\psi^{\prime}(-\xi_2)} + \frac{e^{-\xi_1 x}}{\psi^{\prime}(-\xi_1)} + \frac{e^{\Phi(u)x}}{\psi^{\prime}(\Phi(u))}, \quad \forall x\geq 0,
\end{align}
where $-\xi_1$, $-\xi_2$, and $\Phi(u)$ denotes three roots of $\psi(\lambda)=u$ s.t. $-\xi_2<-c<-\xi_1<0<\Phi(u)$. It is straightforward to check that $W_{\Phi(u)}(x)$ is given by
\begin{align}\label{eq:ScaleF2}
W_{\Phi(u)}(x)=\frac{e^{-(\xi_2+\Phi(u)) x}}{\psi^{\prime}(-\xi_2)} + \frac{e^{-(\xi_1+\Phi(u)) x}}{\psi^{\prime}(-\xi_1)} + \frac{1}{\psi^{\prime}(\Phi(u))}, \quad \forall x\geq 0.
\end{align}
The convexity of $\psi(\lambda)$ implies $\psi^{\prime}(-\xi_2)<0$, $\psi^{\prime}(-\xi_1)<0$ and $\psi^{\prime}(\Phi(u))>0$. Hence, $W_{\Phi(u)}(x)$ is increasing, concave and bounded from above by $1/\psi^{\prime}(\Phi(u))$.
\end{Ex}

The scale functions $W_{\Phi(r)}(x)$ (\ref{eq:ScaleF2}) and $W^{(r)}(x)/W^{(r)\prime}(x)$ are displayed in Figure \ref{fig:scale}. Both functions are increasing in  $x$ and have non-zero and zero values at $x=0$ when $\sigma=0$ ($X$ has bounded variation) and $\sigma\neq 0$ ($X$ has unbounded variation) respectively. Notice that $W^{(r)}(x)/W^{(r)\prime}(x)$ is bounded above by $1/\Phi(r)$.

 \begin{prop}\label{prop:fptbelow2}
For a given $0<a<b$ and $q\geq 0$, we have for all $y\in[a,b]$,
\begin{align}
\mathbb{E}_{\vert y}\big[e^{-q\tau_a^-}\mathbf{1}_{\{\tau_a^-\leq \tau_b^+\}}\big]=&\frac{W^{(q)}(b-y)}{W^{(q)}(b-a)}, \label{eq:fptbelow2}\\
\mathbb{E}_{\vert y}\big[e^{-q\tau_b^+}\mathbf{1}_{\{\tau_b^+\leq \tau_a^-\}}\big]=&Z^{(q)}(b-y)-\frac{Z^{(q)}(b-a)}{W^{(q)}(b-a)}W^{(q)}(b-y).\label{eq:fptabove}
\end{align}
\end{prop}
\begin{proof}
The proof of (\ref{eq:fptbelow2}) follows from the observation that $\tau_a^-<\tau_{\{0\}}$ a.s. and the equivalent between the two events $\{Y_t, t<\tau_{\{0\}}, \mathbb{P}_{\vert y}\}$ and $\{-X_t, t<T_0^+,\mathbb{P}_{-y}\}$. The identity (\ref{eq:fptabove}) is established using the strong Markov property of $Y$, (\ref{eq:fptbelow2}), (\ref{eq:identity2}) along with applying the tower property of conditional expectation, i.e.
\begin{align*}
\mathbb{E}_{\vert y}\big[e^{-q\tau_b^+}\mathbf{1}_{\{\tau_b^+\leq \tau_a^-\}}\big]=&
\mathbb{E}_{\vert y}\big[e^{-q\tau_b^+}\big]-\mathbb{E}_{\vert y}\big[e^{-q\tau_a^-}\mathbf{1}_{\{\tau_a^- < \tau_b^+\}}\big]\mathbb{E}_{\vert a}\big[e^{-q\tau_b^+}\big]. \quad \exit
\end{align*}
 \end{proof}
We define $\mathcal{G}_b(y; p,\alpha,\gamma) = \overline{C}_{\infty}(y,b; p,\alpha) - \gamma $.
\begin{prop}
Following (\ref{eq:perpayoff}),
\begin{eqnarray}
\overline{C}_{\infty}(y,b; \widetilde{p},\widetilde{\alpha})&=& \widetilde{\alpha} Z^{(r)}(b-y) - \frac{\big(\widetilde{p}+r\widetilde{\alpha} W^{(r)}(b)\big)}{W^{(r)\prime}(b)}W^{(r)}(b-y). \label{eq:eq1}
\end{eqnarray}
Moreover, (\ref{eq:perpost2}) becomes
\begin{eqnarray}\label{eq:OSP}
\mathcal{V}_b(y;\widetilde{p},\widetilde{\alpha},\gamma)&=& \sup_{\theta\in\mathcal{T}_{[0,\infty)}} \mathbb{E}_{\vert y}\big[e^{-r\theta}\mathcal{G}_b(Y_{\theta};\widetilde{p},\widetilde{\alpha},\gamma) ; \theta \leq \tau_b^+\big]. \label{eq:eq2}
\end{eqnarray}
Note that we have used the notational convention: $\mathbb{E}[\cdot ; A]=\mathbb{E}[\cdot \mathbf{1}_A]$.
\end{prop}

\begin{proof}
The proof of (\ref{eq:eq1}) follows from applying Proposition \ref{prop:prop1} to (\ref{eq:perpayoff}). The expression for $\mathcal{V}_b$ is obtained by the strong Markov property. First, recall that
\begin{eqnarray*}
\mathbb{E}_{\vert y}\big[e^{-r\tau_b^+}\mathbf{1}_{\{\theta < \tau_b^+\}}\big]=\mathbb{E}_{\vert y}\Big[ \mathbb{E}\big[e^{-r\tau_b^+}\mathbf{1}_{\{\theta<\tau_b^+\}}\big\vert \mathcal{F}_{\theta}\big]\Big]=\mathbb{E}_{\vert y}\Big[ e^{-r\theta} \mathbf{1}_{\{\theta<\tau_b^+\}} \mathbb{E}_{\vert Y_{\theta}}\big[e^{-r\tau_b^+}\big]\Big],
\end{eqnarray*}
where the inner expectation $\mathbb{E}_{\vert Y_{\theta}}\big[e^{-r\tau_b^+}\big]$ is given using (\ref{eq:identity2}) by
\begin{equation}\label{eq:derv2}
\mathbb{E}_{\vert Y_{\theta}}\big[e^{-r\tau_b^+}\big]=Z^{(r)}(b-Y_{\theta})-r\frac{W^{(r)}(b)}{W^{(r)\prime}(b)}W^{(r)}(b-Y_{\theta}).
\end{equation}
Again, by iterated law of conditional expectation and strong Markov property,
\begin{eqnarray*}
\mathbb{E}_{\vert y}\Big[\mathbf{1}_{\{\theta <\tau_b^+\}} \int_{\theta}^{\tau_b^+} e^{-rt}dS_t\Big]=\mathbb{E}_{\vert y}\Big[e^{-r\theta}\mathbf{1}_{\{\theta<\tau_b^+\}}\mathbb{E}_{\vert Y_{\theta}}\Big[\int_0^{\tau_b^+}e^{-rt}dS_t\Big]\Big].
\end{eqnarray*}
Following the identity (\ref{eq:identity1}), the inner expectation is given by
\begin{equation}\label{eq:derv3}
\mathbb{E}_{\vert Y_{\theta}}\Big[\int_0^{\tau_b^+}e^{-rt}dS_t\Big]=\frac{W^{(r)}(b-Y_{\theta})}{W^{(r)\prime}(b)}.
\end{equation}
Putting the two pieces (\ref{eq:derv3}) and  (\ref{eq:derv2}) together leads to $\mathcal{G}_b(Y_{\theta};\widetilde{p},\widetilde{\alpha},\gamma).$
\exit
\end{proof}

In the sequel below we use the shorthand notation $\mathcal{V}_b(y)$ and $\mathcal{G}_b(y)$ for the value function $\mathcal{V}_b(y;\widetilde{p},\widetilde{\alpha},\gamma)$ and the payoff function $\mathcal{G}_b(y;\widetilde{p},\widetilde{\alpha},\gamma)$, respectively.

\section{Solution of the stopping problem (\ref{eq:OSP})}

In this section we discuss method of solution to the stopping problem (\ref{eq:eq2}). We will show that the stopping problem can be reduced to first-passage of drawdown process $Y$ below a fixed level. Our approach is similar to that of proposed by van Moerbeke in \cite{Moerbeke}. Denote by $\mathcal{L}_Y$ an infinitesimal generator of reflected process $Y=S-X$ defined by
\begin{equation}\label{eq:generator}
\begin{split}
\lefteqn{\mathcal{L}_YF(z)=-\mu F^{\prime}(z) +\frac{\sigma^2}{2}F^{\prime \prime}(z)}\\ &+ \int_{-\infty}^0\big[F(z-w)-F(z)+ w\mathbf{1}_{\{-1\leq w<0\}}F^{\prime}(z)\big]\Pi(dw),
\end{split}
\end{equation}
for bounded continuous function $F$, which is twice continuously differentiable, i.e., $F\in C_b^0(\mathbb{R}_+)\cap C^2(\mathbb{R}_+)$, where $F^{\prime}$ and $F^{\prime \prime}$ denote the first and second derivative of $F$. Note that the above generator corresponds to the case where $X$ has paths of unbounded variation with $\sigma>0$. However, when $X$ has paths of bounded variation, we set $\sigma=0$ and the total jumps in the integral is replaced by $ \int_{\{w<-1\}}\big[F(z-w)-F(z)\big]\Pi(dw)$ followed by adjusting the drift of $X$ (\ref{eq:LevyIto}).

To solve the problem (\ref{eq:OSP}), we reduce the optimal stopping rule to the first-passage below a level of drawdown process $Y$. That is, we will show the value function of the optimal stopping (\ref{eq:OSP}) coincides with the function
\begin{equation}\label{eq:optimsol}
\widetilde{\mathcal{V}}_b(y)=
\begin{cases}
\mathcal{G}_b(y), & \textrm{for $y\in[0,h^{\star}]$}\\
\mathcal{G}_b(h^{\star})\frac{W^{(r)}(b-y)}{W^{(r)}(b-h^{\star})}, & \textrm{for $y\in[h^{\star},b]$},
\end{cases}
\end{equation}
where $0<h^{\star} < b$ is defined as the largest root, when exists, of equation (\ref{eq:bigroot}).

The result below gives a condition on the switching cost $\gamma$ for which  the equation (\ref{eq:bigroot}) has a unique positively valued solution $h^{\star}<b$.

\begin{prop}\label{prop:unique}
There exists a unique solution $h^{\star}$ to the equation (\ref{eq:bigroot}) for
\begin{align}\label{eq:ineq1a}
\widetilde{\alpha}\Big(Z^{(r)}(b)-r \frac{W^{(r)}(b)^2}{W^{(r)\prime}(b)}\Big) > \gamma > \widetilde{\alpha}\Big(1-r \frac{W^{(r)}(0)^2}{W^{(r)\prime}(0)}\Big).
\end{align}
\end{prop}
See Lemma 4.3 and Lemma 4.4 in \cite{Kyprianou2007} for the values of $W^{(r)}(0)$ and $W^{(r)\prime}(0)$.

\begin{proof}
The proof is established in two parts. First, we show for a given $r>0$,
\begin{align}\label{eq:ineq1b}
Z^{(r)}(b)-r \frac{W^{(r)}(b)^2}{W^{(r)\prime}(b)} < 1-r \frac{W^{(r)}(0)^2}{W^{(r)\prime}(0)}, \quad \forall \; b\geq 0,
\end{align}
leading to the existence of such $\gamma$. For this purpose, consider the function
\begin{align}\label{eq:FrB}
f_r(b):= Z^{(r)}(b)-r \Big(\frac{W^{(r)}(b)^2}{W^{(r)\prime}(b)} -  \frac{W^{(r)}(0)^2}{W^{(r)\prime}(0)}\Big) -1.
\end{align}

Taking derivative w.r.t $b$ of $f_r(b)$, we obtain after some calculations that
\begin{align*}
\frac{d}{db}f_r(b)=-rW^{(r)}(b)\frac{d}{db}\Big(\frac{W^{(r)}(b)}{W^{(r)\prime}(b)}\Big)<0,
\end{align*}
where the inequality follows on account of (\ref{eq:monotone}), which in turn leading to (\ref{eq:ineq1b}) given that $f_r(0)=0$ and subsequently to (\ref{eq:ineq1a}) as $\widetilde{\alpha}<0$. Next, from (\ref{eq:bigroot}),
\begin{align}\label{eq:funfh}
&\frac{d}{dh}\Big(f(h):=\widetilde{\alpha} Z^{(r)}(b-h)-r \widetilde{\alpha}\frac{\big(W^{(r)}(b-h)\big)^2}{W^{(r)\prime}(b-h)}-\gamma\Big)\\
&\hspace{2cm}=r\widetilde{\alpha}W^{(r)}(b-h)\frac{d}{dx}\Big(\frac{W^{(r)}(x)}{W^{(r)\prime}(x)}\Big)\Big\vert_{x=b-h}<0, \quad \textrm{for} \; 0\leq h\leq b \nonumber
\end{align}
where the inequality sign is due to $\widetilde{\alpha}<0$ and (\ref{eq:monotone}). Uniqueness of solution to (\ref{eq:bigroot}) follows on account of (\ref{eq:ineq1a}) by which we have $f(0)>0$ and $f(b)<0$.  \exit

\end{proof}

\begin{prop}\label{prop:fptbelow}
Let $\tau_h^-$, with $h>0$, be the stopping time (\ref{eq:fptbelow}). Then,
\begin{equation}\label{eq:suph}
\widetilde{\mathcal{V}}_b(y)=\sup_h
\mathbb{E}_{\vert y}\big[e^{-r\tau_{h}^-}\mathcal{G}_b\big(Y_{\tau_{h}^-}\big); \tau_{h}^- \leq \tau_b^+\big].
\end{equation}
\end{prop}
\begin{proof}
Recall that in the absence of positive jumps,  $Y_{\tau_h^-}=h$ a.s. under $\mathbb{P}_{\vert y}$. Thus, on account of Proposition \ref{prop:fptbelow2}, we have
\begin{equation}\label{hy}
J_h(y):=\mathbb{E}_{\vert y} \big[ e^{-r\tau_h^-} \mathcal{G}_b(Y_{\tau_h^-});\tau_h^-\leq \tau_b^+\big]=\mathcal{G}_b(h)\frac{W^{(r)}(b-y)}{W^{(r)}(b-h)}.
\end{equation}
By applying first order Euler condition to the function $h\rightarrow J_h(y)$, we have
\begin{align*}
0&=\frac{\partial}{\partial h} J_h(y)=W^{(r)}(b-y)\frac{\big(\mathcal{G}_b^{\prime}(h)W^{(r)}(b-h)+\mathcal{G}_b(h)W^{(r)\prime}(b-h)\big)}{\big[W^{(r)}(b-h)\big]^2}\\
&=W^{(r)}(b-y) \frac{\big[ -r\widetilde{\alpha} \big[W^{(r)}(b-h)\big]^2 + \widetilde{\alpha} Z^{(r)}(b-h)W^{(r)\prime}(b-h) -\gamma W^{(r)\prime}(b-h) \big]}{\big[W^{(r)}(b-h)\big]^2},
\end{align*}
from which we deduce following Proposition \ref{prop:unique} that $h^{\star}$ uniquely solves the equation (\ref{eq:bigroot}). Further calculation shows that
\begin{align*}
\frac{\partial^2}{\partial h^2} J_{h}(y)\Big\vert_{h=h^{\star}}=& r\widetilde{\alpha}W^{(r)}(b-h^{\star})W^{(r)\prime}(b-h^{\star})\\
&\times\frac{\big( [W^{(r)\prime}(b-h^{\star})]^2 -W^{(r)}(b-h^{\star})W^{(r)\prime\prime}(b-h^{\star})\big)}{(W^{(r)\prime}(b-h^{\star}))^2}\\
=& r\widetilde{\alpha}W^{(r)}(b-h^{\star})W^{(r)\prime}(b-h^{\star}) \frac{d}{dx}\Big(\frac{W^{(r)}(x)}{W^{(r)\prime}(x)}\Big)\Big\vert_{x=b-h^{\star}},
\end{align*}
which by (\ref{eq:monotone}) confirming that $h^{\star}$ maximizes the function $h\rightarrow J_h(y)$.

Furthermore, for $0\leq y\leq h^{\star}$, $\tau_{h^{\star}}^-=0$ a.s. under $\mathbb{P}_{\vert y}$ leading to $\widetilde{\mathcal{V}}_b(y)=\mathcal{G}_b(y)$ on account of $\mathbb{P}_{\vert y}\{\tau_b^+\geq 0\}=1,$ which in turn establishes (\ref{eq:suph}) and (\ref{eq:optimsol}). \exit
\end{proof}

\medskip

We prove the main result on account of the following fact. Necessarily, we assume throughout the remaining that $\gamma\leq 0$ satisfying the constraint (\ref{eq:ineq1a}).
\begin{prop}\label{prop:ass1}
The payoff function $\mathcal{G}_b(y)$ of (\ref{eq:OSP}) satisfies the equation:
\begin{equation}\label{eq:ass1}
\big(\mathcal{L}_Y-r\big)\mathcal{G}_b(y) =r\gamma, \quad \textrm{for all \;$0\leq y\leq b$}.
\end{equation}
\end{prop}

Note that the left-hand side of inequality \eqref{eq:ass1} is well-defined by \eqref{eq:eq1} and the smoothness of the scale function, which is $C^{1}(\mathbb{R}_+)$ when $X$ has paths of finite variation and the L\'evy measure has no atom, and is $C^2(\mathbb{R}_+)$ if $X$ has paths of unbounded variation with $\sigma>0$.
\begin{theo}\label{theo:main}
The value function $\mathcal{V}_b(y;\widetilde{p},\widetilde{\alpha},\gamma)$ of the optimal stopping problem (\ref{eq:OSP}) is given by (\ref{eq:optimsol}) and is obtained at $\tau_{h^{\star}}^-:=\inf\{t>0:Y_t<h^{\star}\}$, i.e.,
\begin{equation}
\mathcal{V}_b(y;\widetilde{p},\widetilde{\alpha},\gamma)=\mathbb{E}_{\vert y}\big[e^{-r\tau_{h^{\star}}^-}\mathcal{G}_b\big(Y_{\tau_{h^{\star}}^-}; \widetilde{p},\widetilde{\alpha},\gamma\big); \tau_{h^{\star}}^- \leq \tau_b^+\big].
\end{equation}
Furthermore, regardless of the regularity of the sample paths of $X$, the value function satisfies both continuous and smooth pasting conditions at the boundary,
\begin{align*}
\mathcal{V}_b(y;\widetilde{p},\widetilde{\alpha},\gamma) &=\mathcal{G}_b(y;\widetilde{p},\widetilde{\alpha},\gamma)  \quad \textrm{at $y=h^{\star}$,} \\
\mathcal{V}_b^{\prime}(y;\widetilde{p},\widetilde{\alpha},\gamma) &= \mathcal{G}_b^{\prime}(y;\widetilde{p},\widetilde{\alpha},\gamma) \quad \textrm{at $y=h^{\star}$}.
\end{align*}
\end{theo}

\begin{prop}\label{prop:variational}
The function $\mathcal{V}_b(y)$ solves uniquely the variational inequality
\begin{align}\label{eq:varineq}
\max\big\{\mathcal{G}_b(y)-\mathcal{V}_b(y),\big(\mathcal{L}_Y-r\big)\mathcal{V}_b(y)\big\}=0, \quad \textrm{for $0\leq y\leq b$.}
\end{align}
\end{prop}

Note that the equation (\ref{eq:varineq}) may be used/extended to numerically solve the finite-maturity counter part of the optimal stopping problem (\ref{eq:OSP}).

The above theorem states an optimal solution to the credit default swaps by exercising the call option at reduced premium rate $\widehat{p}$ and lower default payment $\widehat{\alpha}$ when the reference asset is increasing subject to paying a cost $\gamma$.

\section{Optimality and uniqueness of the solution}
The following results are required to establish the main results of Section 3.
\begin{lem}\label{lem:martingale}
By the strong Markov property, for any $0\leq h<b$ the processes
\begin{align*}
\big\{e^{-u(t\wedge\tau_h^-\wedge \tau_b^+)}W^{(u)}(b-Y_{t\wedge\tau_h^-\wedge \tau_b^+})\big\}_{t\geq 0} \; ,\;\big\{ e^{-u(t\wedge\tau_h^-\wedge \tau_b^+)} Z^{(u)}(b-Y_{t\wedge\tau_h^-\wedge \tau_b^+})\big\}_{t\geq 0},
\end{align*}
are $\mathcal{F}_t-$martingale under the probability measure $\mathbb{P}_{\vert y}$, for $h\leq y<b$.
\end{lem}

\begin{proof}
The proof follows by adapting the approach of \cite{Avram2004} for drawdown L\'evy process. To be more precise, to show the martingale property of the process $\big\{e^{-u(t\wedge\tau_h^-\wedge \tau_b^+)}W^{(u)}(b-Y_{t\wedge\tau_h^-\wedge \tau_b^+})\big\}_{t\geq 0}$, recall that $W^{(u)}(x)=0$ for $x<0$ and the following $\mathbb{P}_{\vert y}-$almost surely equivalence
\begin{equation}\label{eq:identity}
\mathbf{1}_{\{\tau_h^-\leq \tau_b^+\}}=W^{(u)}(b-Y_{\tau_h^-\wedge \tau_b^+})/W^{(u)}(b-h).
\end{equation}
Thus, following the identity (\ref{eq:fptbelow2}), (\ref{eq:identity}) and the strong Markov property,
\begin{align*}
&\mathbb{E}_{\vert y}\Big[e^{-u(\tau_h^-\wedge \tau_b^+)} \frac{W^{(u)}(b-Y_{\tau_h^-\wedge \tau_b^+})}{W^{(u)}(b-h)} \Big\vert \mathcal{F}_t\Big]\\
&\hspace{2cm}=\mathbf{1}_{\{\tau_h^-\wedge \tau_b^+ \geq t\}}e^{-ut}
\mathbb{E}_{\vert Y_t}\Big[e^{-u(\tau_h^-\wedge \tau_b^+)} \frac{W^{(u)}(b-Y_{\tau_h^-\wedge \tau_b^+})}{W^{(u)}(b-h)} \Big]\\
&\hspace{2.75cm}+\mathbf{1}_{\{\tau_h^-\wedge \tau_b^+ < t\}} e^{-u(\tau_h^-\wedge \tau_b^+)} \frac{W^{(u)}(b-Y_{\tau_h^-\wedge \tau_b^+})}{W^{(u)}(b-h)} \\
&\hspace{2cm}=\mathbf{1}_{\{\tau_h^-\wedge \tau_b^+ \geq t\}} e^{-ut}
\frac{W^{(u)}(b-Y_t)}{W^{(u)}(b-h)} \\
&\hspace{2.75cm}+\mathbf{1}_{\{\tau_h^-\wedge \tau_b^+ < t\}} e^{-u(\tau_h^-\wedge \tau_b^+)} \frac{W^{(u)}(b-Y_{\tau_h^-\wedge \tau_b^+})}{W^{(u)}(b-h)} \\
&\hspace{2cm}=e^{-u(t\wedge\tau_h^-\wedge \tau_b^+)} \frac{W^{(u)}(b-Y_{t\wedge\tau_h^-\wedge \tau_b^+})}{W^{(u)}(b-h)}.
\end{align*}
Hence, $\{e^{-u(t\wedge\tau_h^-\wedge \tau_b^+)}W^{(u)}(b-Y_{t\wedge\tau_h^-\wedge \tau_b^+})\}_{t\geq 0}$ is $\mathbb{P}_{\vert y}$ $\mathcal{F}_t-$martingale. Given that
\begin{align*}
\mathbf{1}_{\{\tau_b^+\leq \tau_h^-\}}=Z^{(u)}(b-Y_{\tau_h^-\wedge \tau_b^+})-\frac{Z^{(u)}(b-h)}{W^{(u)}(b-h)}W^{(u)}(b-Y_{\tau_h^-\wedge \tau_b^+}),
\end{align*}
one can show using the identity (\ref{eq:fptabove}) and the strong Markov property that
\begin{align*}
\Big\{e^{-u(t\wedge\tau_h^-\wedge\tau_b^+)}\Big(Z^{(u)}(b-Y_{t\wedge \tau_h^-\wedge \tau_b^+})-\frac{Z^{(u)}(b-h)}{W^{(u)}(b-h)}W^{(u)}(b-Y_{t\wedge \tau_h^-\wedge \tau_b^+})\Big)\Big\}_{t\geq 0},
\end{align*}
is $\mathbb{P}_{\vert y}$ $\mathcal{F}_t-$martingale, and hence so is $\big\{ e^{-u(t\wedge\tau_h^-\wedge \tau_b^+)} Z^{(u)}(b-Y_{t\wedge\tau_h^-\wedge \tau_b^+})\big\}_{t\geq 0}$. \exit
\end{proof}

\begin{prop}\label{prop:martingale2}
For any $0\leq h<b$, $\big\{e^{-r(t\wedge\tau_h^-\wedge\tau_b^+)}\overline{C}_{\infty}(Y_{t\wedge\tau_h^-\wedge\tau_b^+},b; \widetilde{p},\widetilde{\alpha})
\big\}_{t\geq 0}$ is $\mathcal{F}_t-$martingale under the measure $\mathbb{P}_{\vert y}$, with $h\leq y\leq b$.
\end{prop}
\begin{proof}
The proof is straightforward from applying Lemma \ref{lem:martingale} to (\ref{eq:eq1}). \exit
\end{proof}

\begin{prop}\label{prop:propofv}
For all $y\in[h^{\star},b]$, the function $\widetilde{\mathcal{V}}_b(y)$ satisfies:
\begin{enumerate}
\item[(i)] $\widetilde{\mathcal{V}}_b^{\prime}(y)\leq 0$ and $\widetilde{\mathcal{V}}_b(y)\geq 0$ (for all $0\leq y\leq b$),
\item[(ii)] $(\mathcal{L}_Y-r)\widetilde{\mathcal{V}}_b(y) = 0,$
\item[(iii)] $\widetilde{\mathcal{V}}_b(y) \geq \mathcal{G}_b(y)$.
\end{enumerate}
\end{prop}
\begin{proof}
\begin{enumerate}
\item[(i)] The proof is straightforward following the definition of $\widetilde{\mathcal{V}}_b(y)$ (\ref{eq:optimsol}), (\ref{eq:monotone}), and the fact that $0<W^{(r)}(x)$, increasing $\forall x\geq 0$ and for $\widetilde{\alpha},\widetilde{p}<0$,
\begin{equation*}\label{eq:positiveg}
\begin{split}
\mathcal{G}_b(h^{\star})=r\widetilde{\alpha}W^{(r)}(b-h^{\star})\Big[\frac{W^{(r)}(b-h^{\star})}{W^{(r)\prime}(b-h^{\star})} - \frac{W^{(r)}(b)}{W^{(r)\prime}(b)}\Big] - \widetilde{p} \frac{W^{(r)}(b-h^{\star})}{W^{(r)\prime}(b)} >0
\end{split}
\end{equation*}
and the payoff function $\mathcal{G}_b(y)$ is monotone decreasing for all $0\leq y\leq b$ as
\begin{align}\label{eq:dervG}
\mathcal{G}_b^{\prime}(y)=\widetilde{p}\frac{W^{(r)\prime}(b-y)}{W^{(r)\prime}(b)}
+r\widetilde{\alpha}\frac{\big(W^{(r)}(b)\big)^2}{W^{(r)\prime}(b)}\frac{d}{db}\Big(\frac{W^{(r)}(b-y)}{W^{(r)}(b)}\Big) \leq 0. \quad \exit
\end{align}

\item[(ii)] By Lemma \ref{lem:martingale}, $\big\{e^{-r(t\wedge\tau_{h^{\star}}^-\wedge\tau_b^+)} \widetilde{\mathcal{V}}_b(Y_{ t \wedge \tau_{h^{\star}}^-\wedge\tau_b^+})\big\}_{t\geq 0}$ is $\mathcal{F}_t-$martingale. Hence, on account that the event $\{t:t<\tau_{\{0\}}\}$ has zero Stieltjes measure $dS_t$ under $\mathbb{P}_{\vert y}$, it implies that $(\mathcal{L}_Y-r) \widetilde{\mathcal{V}}_b(y)=0$ for all $y\in[h^{\star},b]$, see (\ref{eq:Ito}).

\item[(iii)]
The proof follows from definition of $h^{\star}$ (\ref{eq:bigroot}) and (\ref{eq:monotone}) by which we have
\begin{align*}
\widetilde{\mathcal{V}}_b^{\prime}(y)-\mathcal{G}_b^{\prime}(y)=r\widetilde{\alpha}W^{(r)\prime}(b-y)\Big[\frac{W^{(r)}(b-y)}{W^{(r)\prime}(b-y)} - \frac{W^{(r)}(b-h^{\star})}{W^{(r)\prime}(b-h^{\star})}\Big]\geq 0.
\end{align*}
The claim on the majorant property follows as $\widetilde{\mathcal{V}}_b(h^{\star}) - \mathcal{G}_b(h^{\star})=0$. \exit
\end{enumerate}
\end{proof}

\begin{prop}\label{prop:CR}
The process $\big\{e^{-r(t\wedge \tau_b^+)}\widetilde{\mathcal{V}}_b(Y_{t\wedge\tau_b^+})\big\}_{t\geq 0}$ is $\mathcal{F}_t-$supermartingale.
\end{prop}
\begin{proof}
Given the smoothness of the scale function $W^{(r)}(x)$, we have by applying the change-of-variable formula for the trivariate process $(t,S_t,X_t)$, see Theorem 33 in Protter \cite{Protter}, applied to the discounted process $e^{-r(t\wedge\tau_b^+)}\widetilde{\mathcal{V}}_b(Y_{t\wedge\tau_b^+})$, the L\'evy-It\^o sample paths decomposition of the discounted process given for $t\geq 0$ by
\begin{equation}\label{eq:Ito}
\begin{split}
e^{-r(t\wedge \tau_b^+)}\widetilde{\mathcal{V}}_b(Y_{t\wedge \tau_b^+)})=& \widetilde{\mathcal{V}}_b(y) + \int_0^{t\wedge \tau_b^+} e^{-ru} \widetilde{\mathcal{V}}_b^{\prime}(0)\mathbf{1}_{\{Y_u=0\}}dS_u \\
&\hspace{1.5cm}+ \int_0^{t\wedge \tau_b^+} e^{-ru}\big(\mathcal{L}_Y -r)\widetilde{\mathcal{V}}_b(Y_u)du + M_{t\wedge \tau_b^+},
\end{split}
\end{equation}
under $\mathbb{P}_{\vert y}$, with $0\leq y\leq b$, where by Doob's optional stopping theorem, $M_{t\wedge \tau_b^+}$ is $\mathcal{F}_t-$martingale with $\mathbb{E}_{\vert y}\big[M_{t\wedge \tau_b^+}\big]=0$. By $(ii)$ of Proposition \ref{prop:propofv} and that $\widetilde{\mathcal{V}}_b(y)=\mathcal{G}_b(y)$ for $0\leq y\leq h^{\star}$, the claim is established on account of (\ref{eq:ass1})
and (\ref{eq:dervG}) by which it follows by definition (\ref{eq:optimsol}) of $\widetilde{\mathcal{V}}_b(y)$ that $\widetilde{\mathcal{V}}_b^{\prime}(0)=\mathcal{G}_b^{\prime}(0)\leq 0$. \exit
\end{proof}

\subsection{Proof of Proposition \ref{prop:ass1}}
On account of the fact that the event $\{t<\tau_{h}^-\wedge\tau_b^+\}$ has zero Stieltjes measure $dS_t$ under $\mathbb{P}_{\vert y}$, with $0\leq h\leq y\leq b$, it follows from Proposition \ref{prop:martingale2} and the paths decomposition (\ref{eq:Ito}) for the process $\big\{e^{-r(t\wedge\tau_h^-\wedge\tau_b^+)}\overline{C}_{\infty}(Y_{t\wedge\tau_h^-\wedge\tau_b^+},b; \widetilde{p},\widetilde{\alpha})\big\}_{t\geq 0}$ that $(\mathcal{L}_Y-r)\mathcal{G}_b(y)=r\gamma$ for all $y\in[h,b]$. The claim (\ref{eq:ass1}) is established given that $h$ is arbitrary.\exit

\subsection{Proof of Theorem \ref{theo:main}}
Recall following (\ref{eq:eq2}) that the value function $\mathcal{V}_b(y)$ satisfies the majorant property over the payoff function $\mathcal{G}_b(y)$, i.e., $\mathcal{V}_b(y)\geq\mathcal{G}_b(y)$ for all $y\in[0,b]$. More precisely, following (\ref{eq:OSP}) we have $\mathcal{V}_b(y)\geq \mathbb{E}_{\vert y}\big[e^{-r\theta} \mathcal{G}_b(Y_{\theta});\theta\leq \tau_b^+\big]$ for all stopping time $\theta\in\mathcal{T}_{[0,\infty)}$. Since $0\in\mathcal{T}_{[0,\infty)}$ and $\mathbb{P}_{\vert y}\{\tau_b^+\geq 0\}=1$, the claim follows for $\theta=0$. Moreover, $\mathcal{V}_b(y)=\mathcal{G}_b(y)$ holds for some $0\leq y \leq b$ such that $\mathbb{P}_{\vert y}\{\theta=0\}=1$. The set $\mathcal{S}=\{0\leq y \leq b: \mathcal{V}_b(y)=\mathcal{G}_b(y)\}$ corresponds to the stopping region of the problem (\ref{eq:eq2}). If the value function $\mathcal{V}_b(y)$ is continuous, $\mathcal{S}$ is a closed set. The complement $\mathcal{C}$ of the set $\mathcal{S}$ refers to the continuation region of (\ref{eq:eq2}), i.e., $\mathcal{C}=\{0\leq y \leq b: \mathcal{V}_b(y)>\mathcal{G}_b(y)\}$. To show that the stopping problem (\ref{eq:OSP}) can be reduced under (\ref{eq:ass1}) to the first-passage below a level of drawdown process $Y$, let us rewrite without loss of generality the problem (\ref{eq:OSP}) as follows:
\begin{equation*}
\mathcal{V}_b(y)=\sup_{\theta\in\mathcal{T}_{[0,\tau_b^+)}}\mathbb{E}_{\vert y}\big[e^{-r\theta}\mathcal{G}_b(Y_{\theta})\big].
\end{equation*}
Note following (\ref{eq:OSP}) that $\mathcal{V}_b(y)\geq \widetilde{\mathcal{V}}_b(y)$ for all $y\in[0,b)$. Moreover, we have $Y_t\in [0,b]$ for $t\leq \tau_b^+$.
To show the reverse inequality we will use Optional Stopping Theorem.

First, we check that the continuous pasting is satisfied at the boundary $y=h^{\star}$. To show the smooth pasting condition, recall that the right derivative of the function $\mathcal{V}_b(y)=\mathcal{G}_b(h^{\star})W^{(r)}(b-y)/W^{(r)}(b-h^{\star})$ at the point $y=h^{\star}$ is given by $\mathcal{V}_b^{\prime}(h^{\star})=-\mathcal{G}_b(h^{\star})W^{(r)\prime}(b-h^{\star})/W^{(r)}(b-h^{\star})$. By evaluating the latter on account of the fact that $h^{\star}$ solves the equation (\ref{eq:bigroot}) leads to
\begin{equation*}
\mathcal{V}_b^{\prime}(h^{\star})=-r\widetilde{\alpha} W^{(r)}(b-h^{\star})+\frac{\big(\widetilde{p}+r\widetilde{\alpha}W^{(r)}(b)\big)}{W^{(r)\prime}(b)}W^{(r)\prime}(b-h^{\star})=\mathcal{G}_b^{\prime}(h^{\star}). \quad
\end{equation*}
As a result, we see that the continuous and smooth pasting conditions are satisfied regardless of the regularity condition on the sample paths of the L\'evy process.

Recall following Proposition \ref{prop:CR} that the process $\{e^{-r(t\wedge\tau_b^+)}\widetilde{\mathcal{V}}_b(Y_{t \wedge\tau_b^+}), t\geq 0\}$ is supermartingale. Hence, by Lemma 7 of Palmowski and Tumilewicz \cite{Palmowski2018}, positivity of $\widetilde{\mathcal{V}}_b(y)$, and $(iii)$ of Proposition \ref{prop:propofv} we have for any stopping time $\theta$,
\begin{align}
\widetilde{\mathcal{V}}_b(y)&
\geq \mathbb{E}_{\vert y}\big[e^{-r(\theta\wedge\tau_b^+)}\widetilde{\mathcal{V}}_b(Y_{\theta\wedge\tau_b^+})\big]  \geq  \mathbb{E}_{\vert y}\big[e^{-r\theta} \mathcal{G}_b(Y_{\theta}); \theta \leq \tau_b^+\big].  \label{eq:ineq2}
\end{align}

We used above the majorant property $\widetilde{\mathcal{V}}_b(y)\geq \mathcal{G}_b(y)$ that holds $\forall y\in[0,b]$.
Taking supremum over all stopping times on the right hand side of \eqref{eq:ineq2}
completes the proof of the first assertion.
\exit

\subsection{Proof of Proposition \ref{prop:variational}}
It is straightforward to check following Propositions \ref{prop:propofv} and \ref{prop:ass1} that the value function $\mathcal{V}_b(y)$ of the optimal stopping (\ref{eq:OSP}) satisfies the variational inequality (\ref{eq:varineq}). Let $(U,d)$ be a pair solution to (\ref{eq:varineq}) such that $U(y)\geq 0$ for all $0\leq y\leq b$, $U(y)=\mathcal{G}_b(y)$ for all $0\leq y\leq d$ and $U(y)\geq \mathcal{G}_b(y)$, otherwise. Assume that $U$ has degree of smoothness such that the L\'evy-It\^o decomposition (\ref{eq:Ito}) applies, i.e.,
\begin{equation}\label{eq:Ito2}
\begin{split}
e^{-r(t\wedge\tau_b^+)}U(Y_{t\wedge\tau_b^+})=&U(y)+\int_0^{t\wedge \tau_b^+} e^{-ru} U^{\prime}(0)\mathbf{1}_{\{Y_u=0\}}dS_u\\
&+r\gamma\int_0^{t\wedge\tau_b^+}e^{-ru}\mathbf{1}_{\{0\leq Y_u\leq d\}}du + M_{t\wedge\tau_b^+}.
\end{split}
\end{equation}
Notice that we have applied the result of Proposition \ref{prop:ass1}. Since $(\mathcal{V}_b,h^{\star})$ is a pair of optimal solution to the stopping problem (\ref{eq:OSP}), we have for all $0\leq y\leq b$,
\begin{align}\label{eq:ineq3}
\mathcal{V}_b(y)\geq U(y),
\end{align}
which in turn implies that $h^{\star}\leq d$. Next, for a given $y\in[h^{\star},d]$, we have after replacing $t$ by $\tau_{h^{\star}}^-$ in the decomposition (\ref{eq:Ito2}) and taking expectation $\mathbb{E}_{\vert y}$ that
\begin{align}\label{eq:ineq4}
\mathbb{E}_{\vert y}\big[e^{-r(\tau_{h^{\star}}^-\wedge\tau_b^+)}U(Y_{\tau_{h^{\star}}^-\wedge\tau_b^+})\big]=U(y) + r\gamma\mathbb{E}_{\vert y}\Big[\int_0^{\tau_{h^{\star}}^-\wedge\tau_b^+}e^{-ru}\mathbf{1}_{\{0\leq Y_u\leq d\}}du\Big].
\end{align}
By positivity of $U(y)$ on $0\leq y\leq b$ and that $U(y)=\mathcal{G}_b(y)$ for $0\leq y\leq d$ along with the fact that $\mathbb{E}_{\vert y}\big[e^{-r\tau_{h^{\star}}^-}\mathcal{G}_b(Y_{\tau_{h^{\star}}^-});\tau_{h^{\star}}^-\leq \tau_b^+\big]=\mathcal{V}_b(y),$ we then obtain
\begin{align*}
U(y) + r\gamma\mathbb{E}_{\vert y}\Big[\int_0^{\tau_{h^{\star}}^-\wedge\tau_b^+}e^{-ru}\mathbf{1}_{\{0\leq Y_u\leq d\}}du\Big]\geq \mathcal{V}_b(y),
\end{align*}
which by (\ref{eq:ineq3}) leads to contrary given that $\gamma\leq 0$. Hence, $\{h^{\star}\leq d\}$ is an empty set which in turn it follows that $h^{\star}=d$ and $U(y)=\mathcal{V}_b(y)$ for all $0\leq y\leq b$.

Similar arguments may be adapted to deal with a finite maturity counterpart of the stopping problem (\ref{eq:OSP}) for which case the proof of unique solution to (\ref{eq:varineq}) is reduced to showing uniqueness of curved stopping boundary $h^{\star}(t)$ solving nonlinear integral equation (\ref{eq:ineq4}). This approach was used in Jacka \cite{Jacka}, Peskir \cite{Peskir} and Surya \cite{Surya2007} for the case of pricing American put option. \exit

\section{Numerical examples}

To exemplify the main results, we discuss some numerical examples for one-sided jump-diffusion process $X$ with Laplace exponent $\psi(\lambda)=\mu\lambda+\frac{\sigma^2}{2}\lambda^2-\frac{a\lambda}{\lambda+c}$ for all $\lambda\in\mathbb{R}$ s.t. $\lambda\neq -c$. See Example \ref{ex:Ex1} for the corresponding scale function. We set $\mu=0.075$, $a=0.5$ and $c=9$ (on average once every two years the firm suffers an instantaneous loss of $10\%$ of its value). We assume that the firm's default level is $b=\log(5)$ and $r=10\%$. The issuer calls the existing contract with a new one offering $\widetilde{\alpha}=-\$5$ less default coverage with lower premium rate $\widetilde{p}=-2.5\%$ than the existing credit default swap, subject to the switching cost $\gamma=-\$1$.

\begin{figure}[ht!]
\centering
\begin{tabular}{cc}
%\hline
\subf{\includegraphics[width=72.5mm]{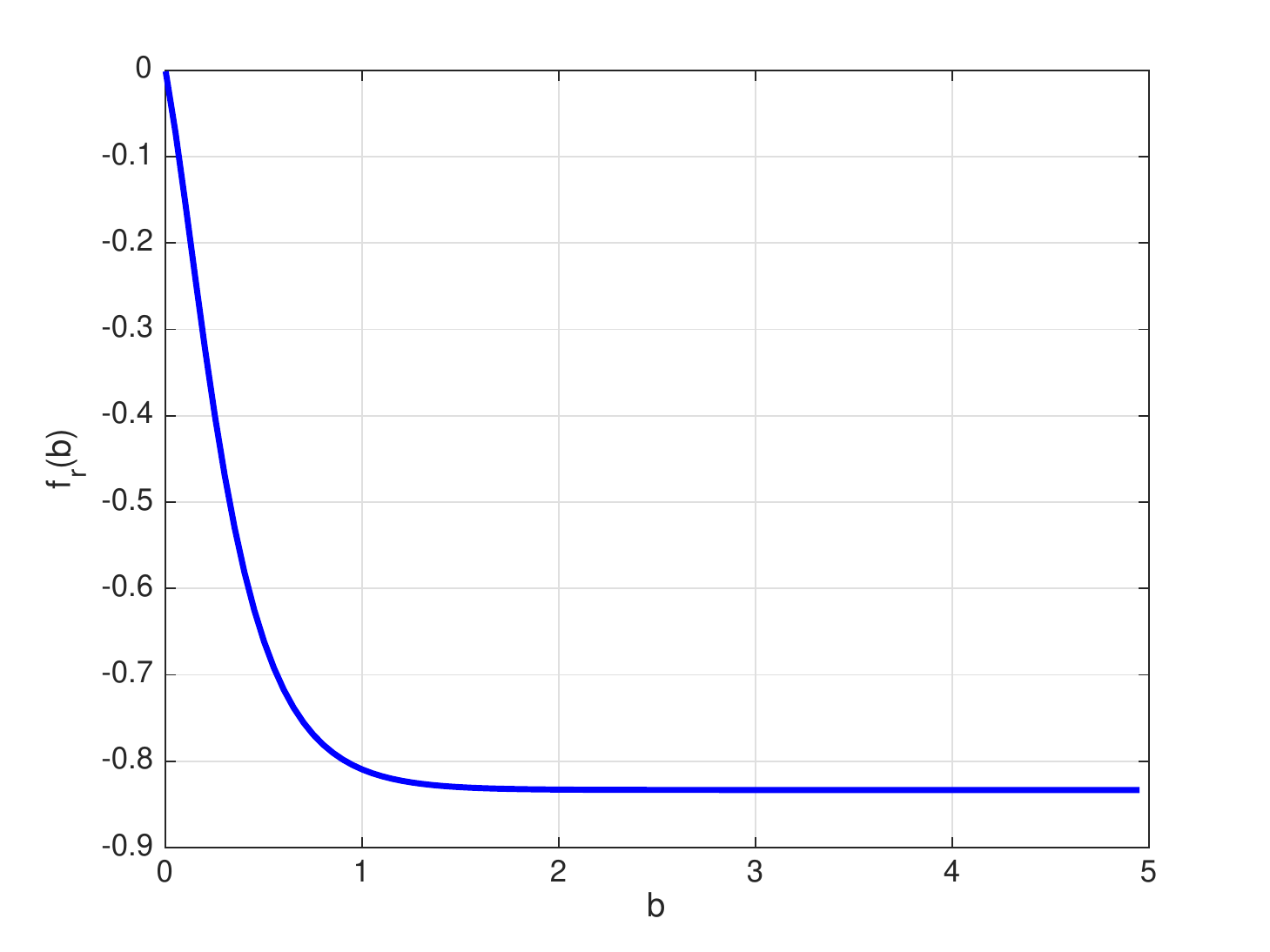}}
       {the case $\sigma=0$.}

&
\subf{\includegraphics[width=72.5mm]{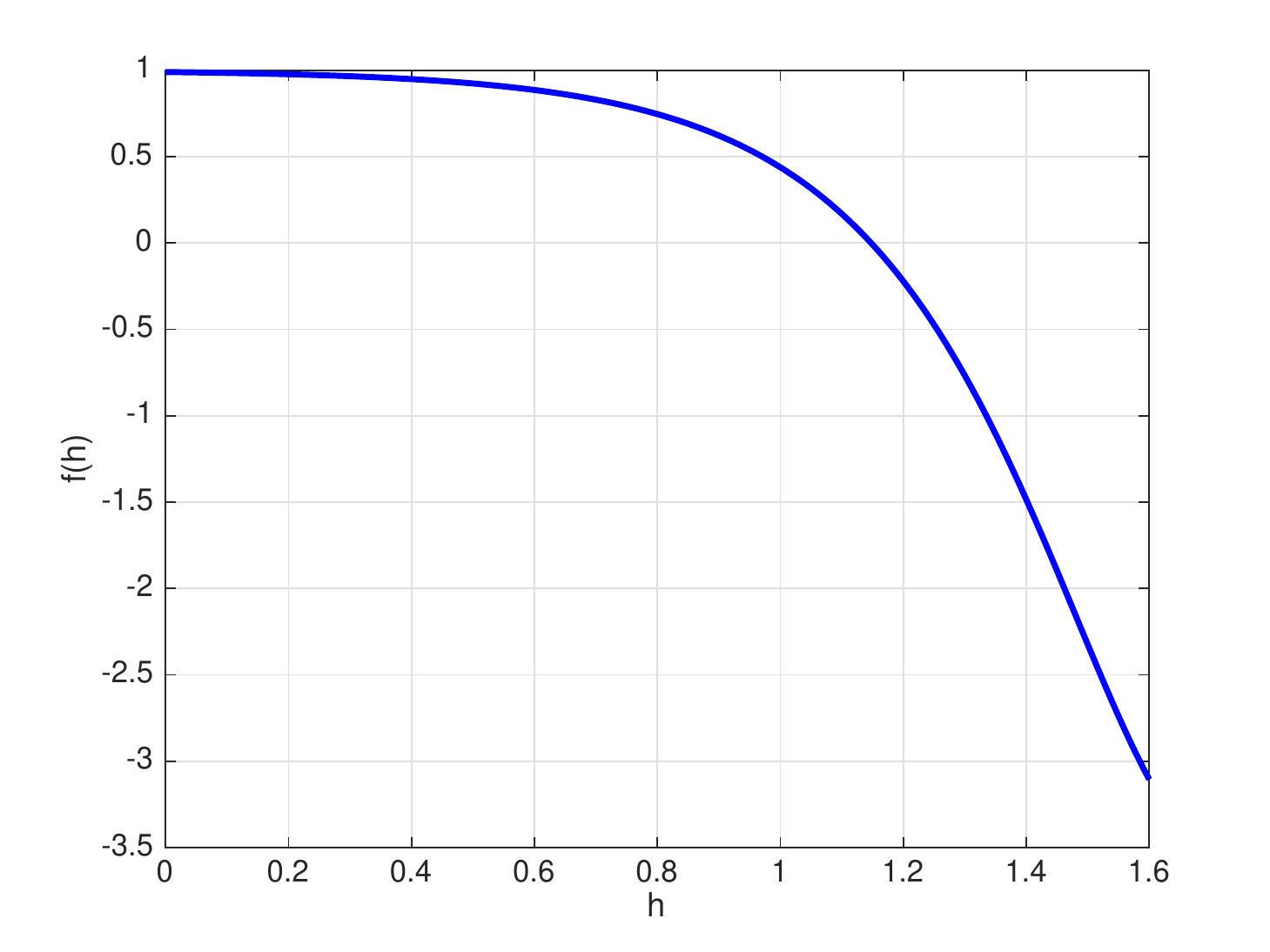}}
     {the case $\sigma=0$.}
\\
%\hline
\subf{\includegraphics[width=72.5mm]{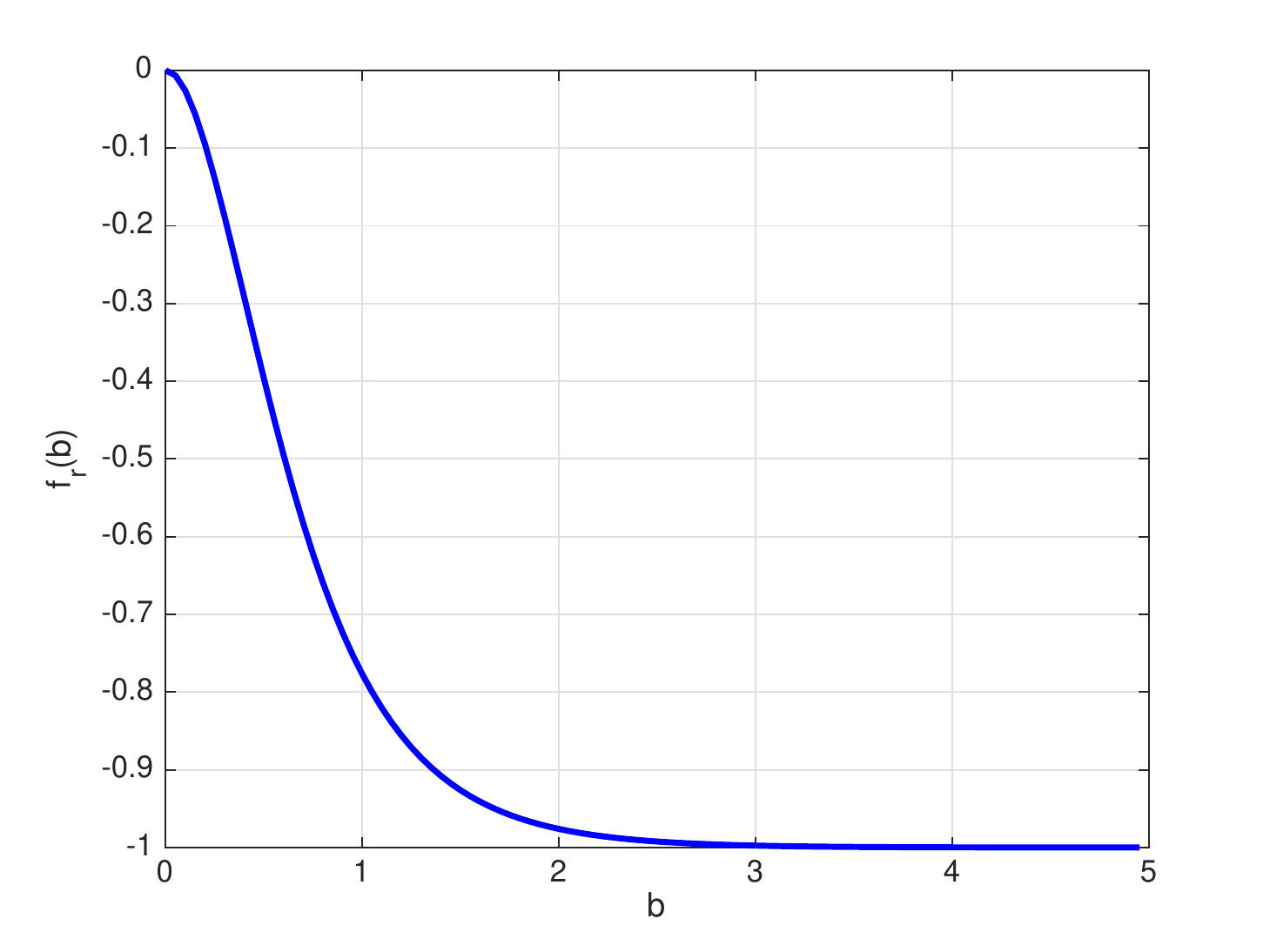}}
       {the case $\sigma=0.2$.}
&
\subf{\includegraphics[width=72.5mm]{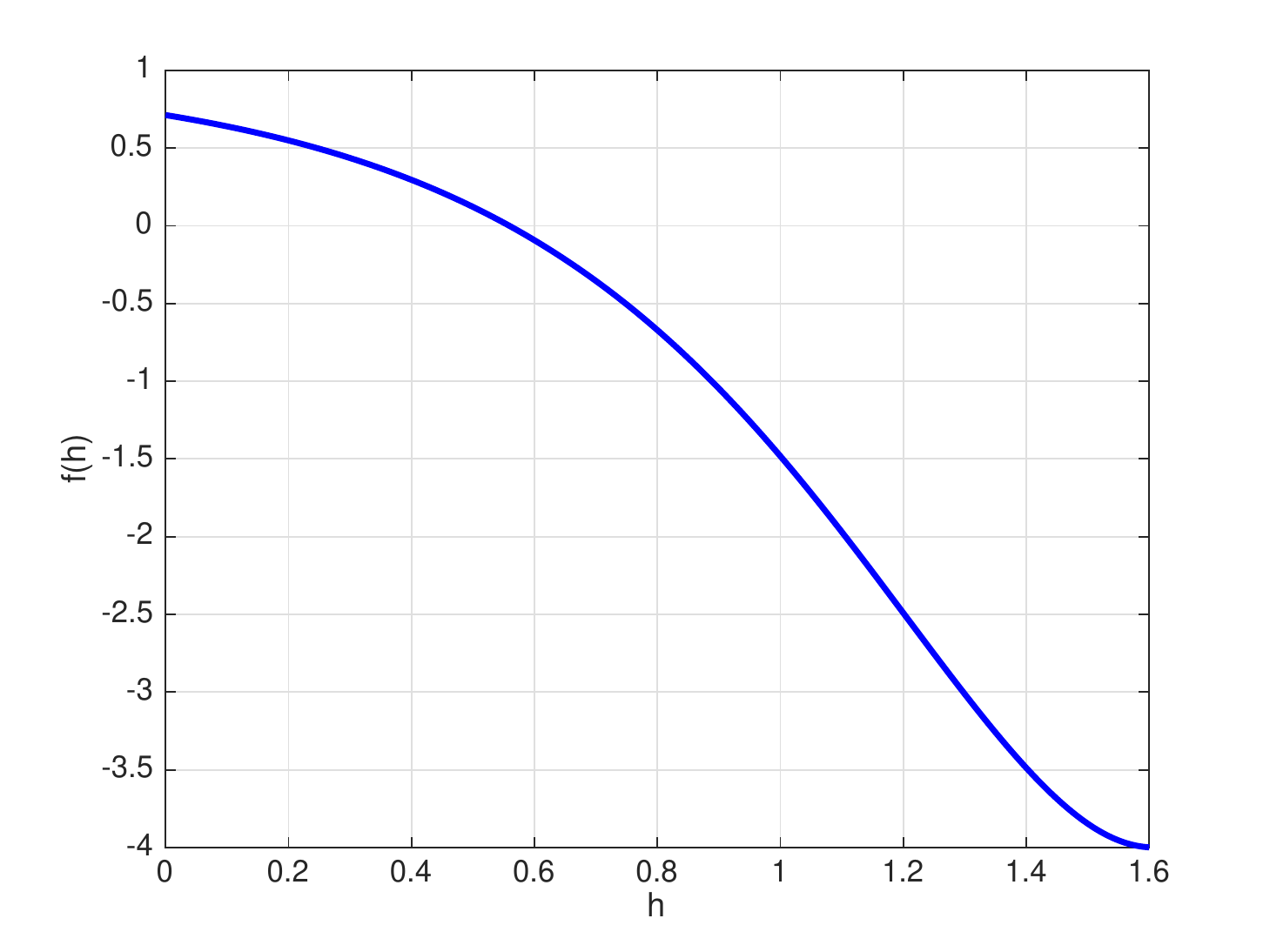}}
     {the case $\sigma=0.2$.}
\\
%\hline
\end{tabular}
\caption{Plots of the functions $f_r(b)$ (\ref{eq:FrB}) and $f(h)$, the left hand side of (\ref{eq:bigroot}).}\label{fig:gam1}
\end{figure}

We consider two cases: $\sigma=0$ and $\sigma=0.2$. The first case corresponds to the underlying process $X$ of the firm value having paths of bounded variation, whereas the other with unbounded variation. Figure \ref{fig:gam1} displays the function $f_r(b)$ (\ref{eq:FrB}) introduced in the proof of uniqueness of the solution to eqn. (\ref{eq:bigroot}), and the function $f(h)$, the left hand side of (\ref{eq:bigroot}). In both cases we notice that the two functions exhibit decreasing property which is required in the proof, in particular the function $f(h)$ has a unique root $h=h^{\star}$ below which $f(h)$ is negative.

\begin{figure}[ht!]
\centering
\begin{tabular}{cc}
%\hline
\subf{\includegraphics[width=72.5mm]{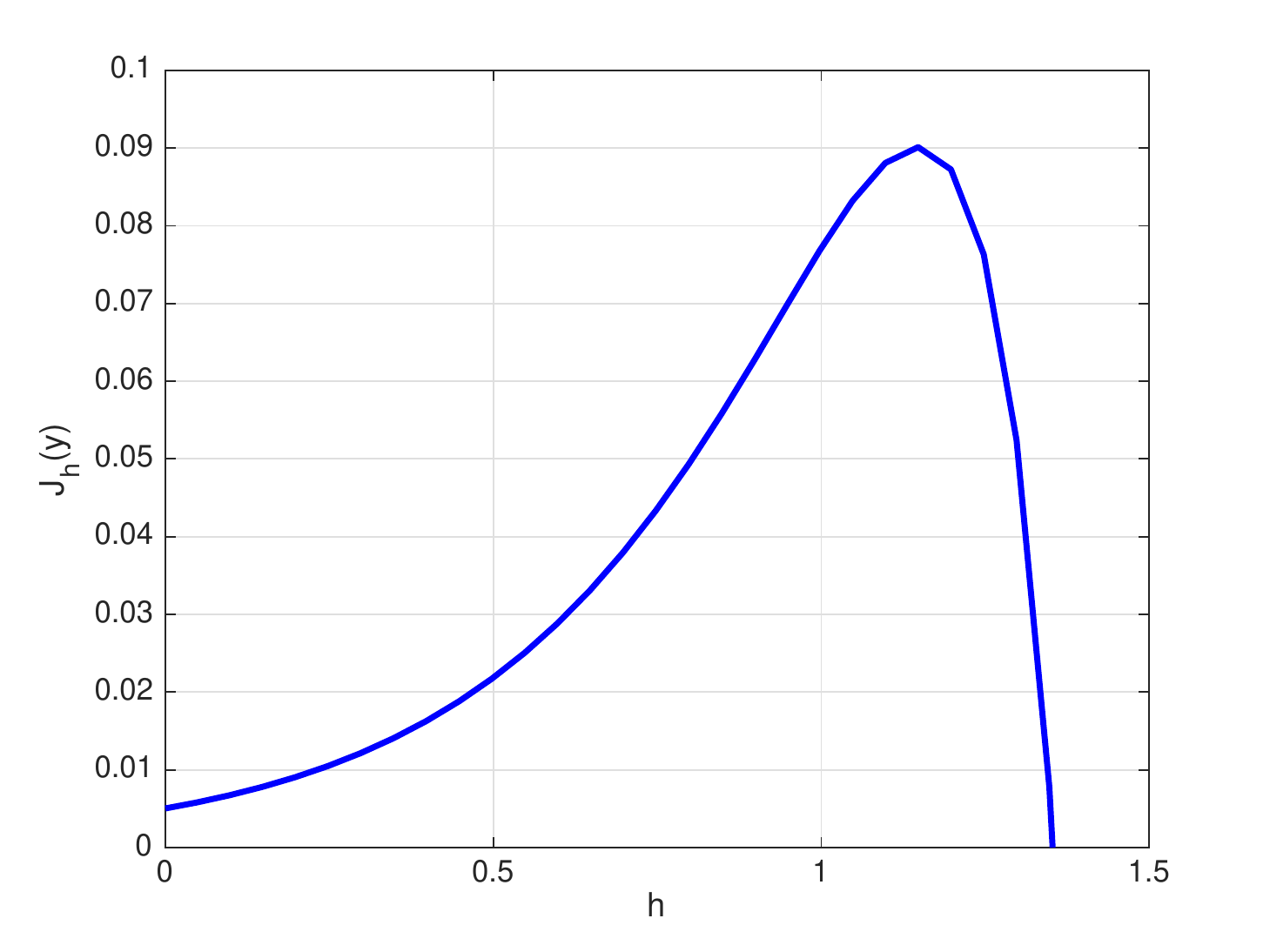}}
       {the case $\sigma=0$.}

&
\subf{\includegraphics[width=72.5mm]{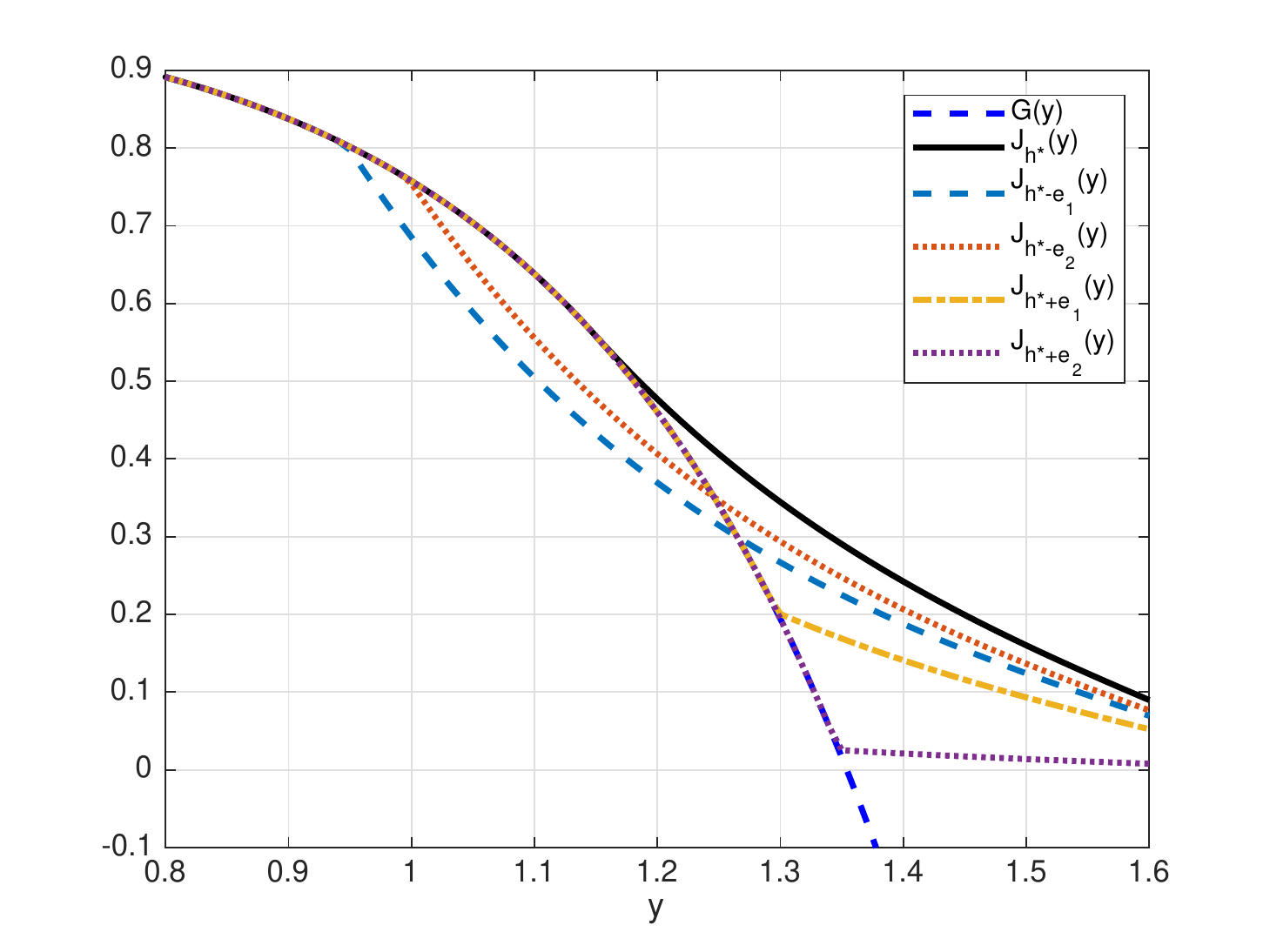}}
     {the case $\sigma=0$.}
\\
%\hline
\subf{\includegraphics[width=72.5mm]{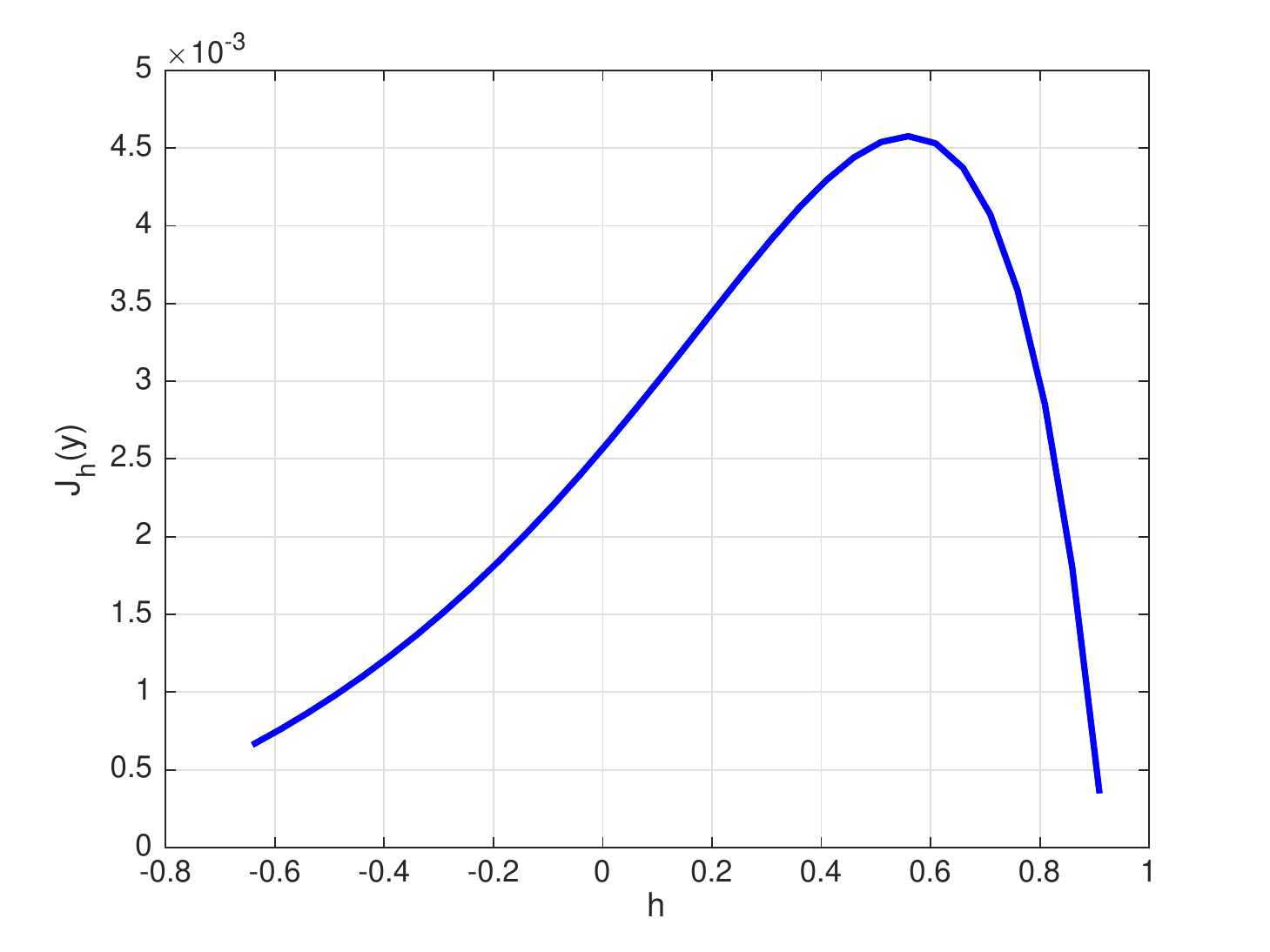}}
       {the case $\sigma=0.2$.}
&
\subf{\includegraphics[width=72.5mm]{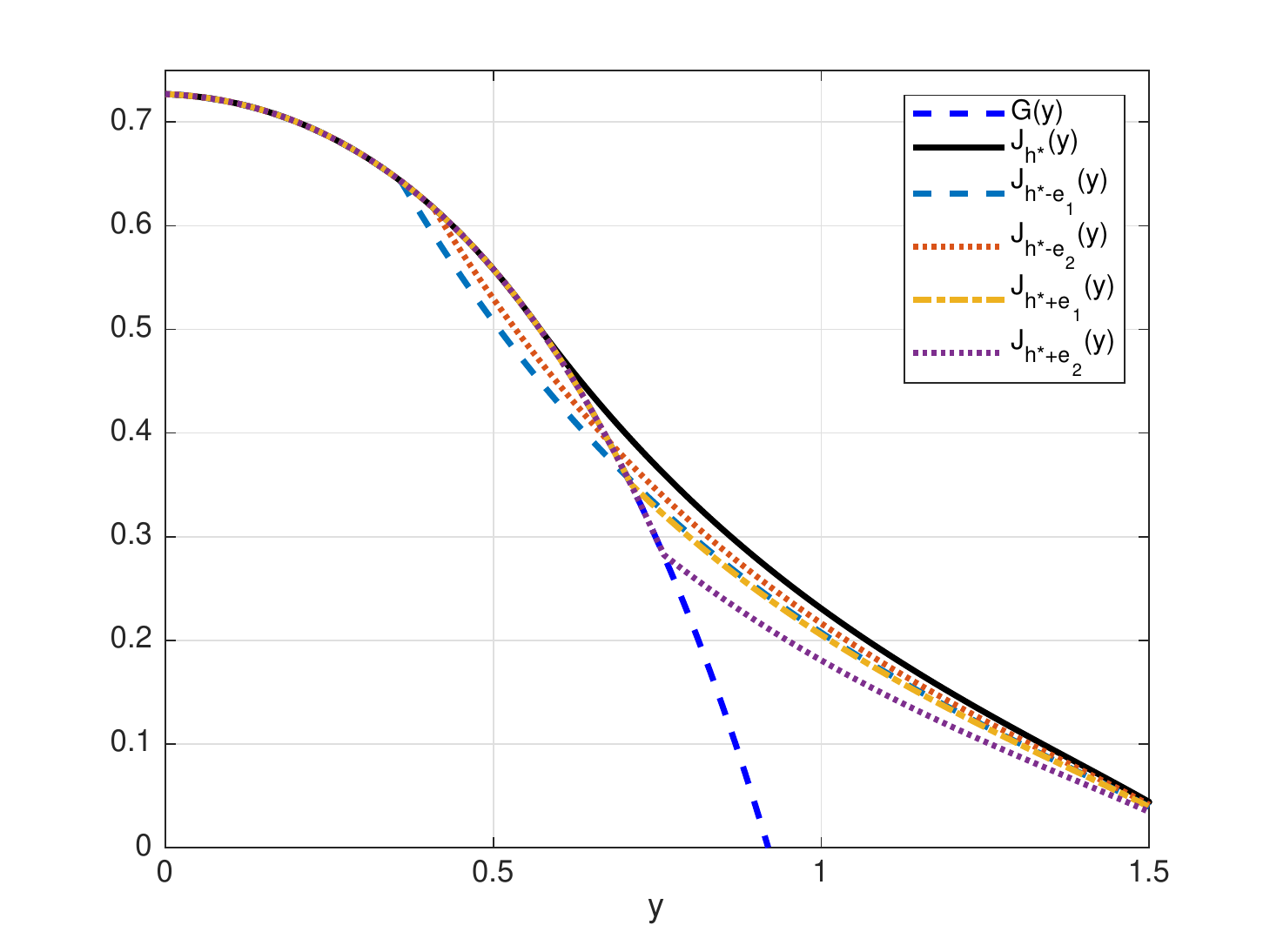}}
     {the case $\sigma=0.2$.}
\\
%\hline
\end{tabular}
\caption{The function $h\rightarrow J_h(y)$ and the value function of (\ref{eq:OSP}).}\label{fig:gam2}
\end{figure}

Figure \ref{fig:gam2} presents the shape of the function $J_h(y)$ (\ref{hy}) for various values of $h<y$ and the value function $\widetilde{\mathcal{V}}_b(y)=J_{h^{\star}}(y)$. The optimal stopping level for the case $\sigma=0$ is $h^{\star}=1.1476$, whereas $h^{\star}=0.5590$ for $\sigma=0.2$. In both cases, we see that the function $h\rightarrow J_h(y)$ achieves the maximum value at $h=h^{\star}$. As a result, the value function $\mathcal{V}_{b}(y)$ of the optimal stopping (\ref{eq:OSP}) dominates the pay-off function $\mathcal{G}_b(y)$ and sub-optimal solution $J_{h^{\star}\pm\varepsilon}(y)$, for $\varepsilon>0$, of the stopping problem (\ref{eq:OSP}) for all values of $0\leq y\leq b$. The value function $\widetilde{\mathcal{V}}_b(y)$ is positively valued and is decreasing. Unlike sub-optimal solutions $J_{h^{\star}\pm\varepsilon}(y)$, $\widetilde{\mathcal{V}}_b(y)$ satisfies both continuous and smooth pasting conditions at the optimal boundary $h^{\star}$ for both cases $\sigma=0$ and $\sigma=0.2$, all of which confirm the main results.

Applying the infinitesimal generator $\mathcal{L}_Y$ (\ref{eq:generator}) to the payoff function $\mathcal{G}_b(y)$, expressed in terms of the scale function (\ref{eq:ScaleF}), the function $(\mathcal{L}_Y-r)\mathcal{G}_b(y)$ is plotted for all $0\leq y\leq b$ in Figure \ref{fig:gam3}. The graph shows that the function has the same value $-0.1$ for all $y\in[0,b]$, which is indeed equal to $r\gamma$ (\ref{eq:ass1}).

\begin{figure}[ht!]
\centering
\begin{tabular}{cc}
%\hline
\subf{\includegraphics[width=72.5mm]{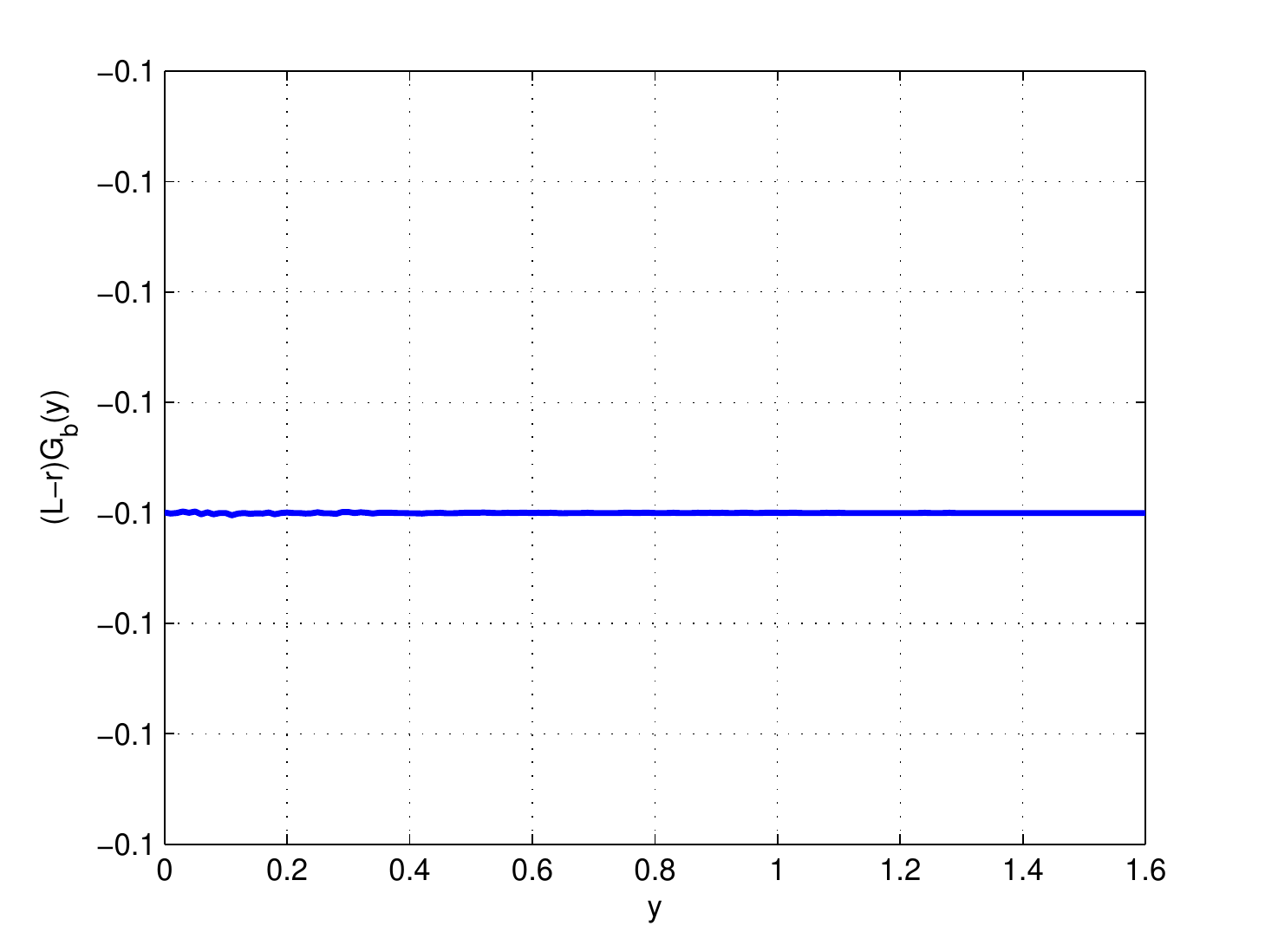}}
       {the case $\sigma=0$.}

&
\subf{\includegraphics[width=72.5mm]{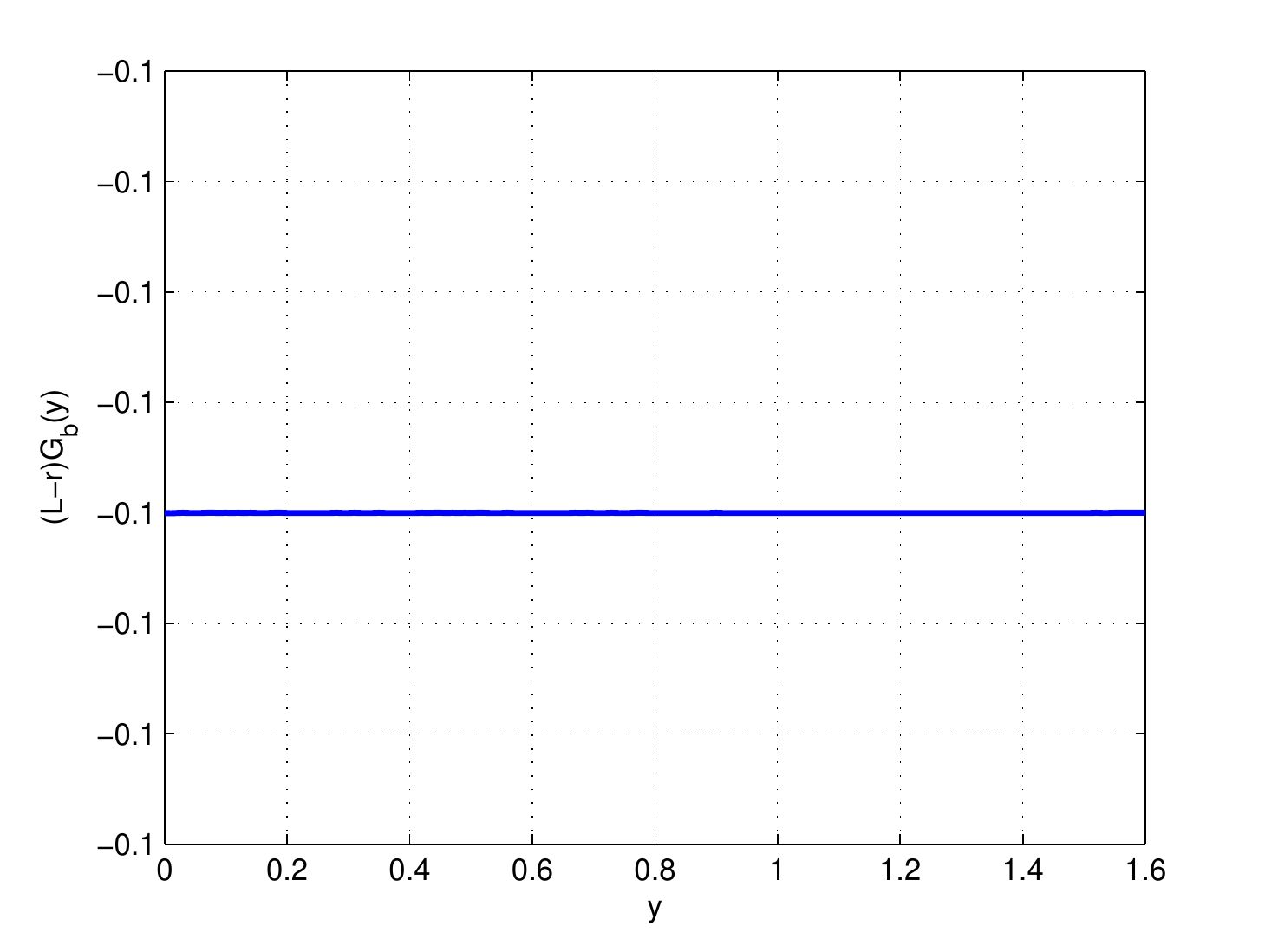}}
     {the case $\sigma=0.2$.}
\\
%\hline
\end{tabular}
\caption{The function $(\mathcal{L}_Y-r)\mathcal{G}_b(y)$ for $0\leq y\leq b$.}\label{fig:gam3}
\end{figure}

\section{Conclusion}
This paper presents optimal valuation of American call option for credit default swaps under drawdown of L\'evy process with only downward jumps. The option gives an opportunity for the issuer to call back the existing swaps contract by replacing it with a new one at reduced premium payment rate with slightly lower default coverage subject to paying some costs. The valuation is formulated in terms of optimal stopping which a risk-neutral protection buyer solves over a class of stopping times adapted to the natural filtration of the asset price. Solution to the optimal stopping problem exists under some constraints imposed on the new premium rate, default coverage and the costs to call the contract. The solution is given explicitly in terms of the scale function of the L\'evy process. Optimality and uniqueness of the solution are established using martingale approach for drawdown and convexity of the scale function under Esscher transform of measure. Numerical examples are presented to confirm the main results that the solution of the stopping problem (the fair value of the call option) is positively valued, decreasing and
has majorant property over the payoff function. Furthermore, it satisfies
both continuous and smooth pasting conditions which holds regardless of the regularity of the sample paths of the L\'evy process.

\end{document}